\documentclass[a4paper,10pt]{article}
\usepackage{graphicx}
\usepackage{amsmath}
\usepackage{citesort}
\usepackage{amsfonts}
\usepackage{amssymb}
\usepackage{amsmath}
\usepackage{latexsym}
\usepackage{srcltx}
\usepackage{color}
\usepackage{ntheorem}
\usepackage{braket}
\usepackage{url}
\usepackage[multiple]{footmisc}

\newcommand{\mnote}[1]
{\protect{\stepcounter{mnotecount}}$^{\mbox{\footnotesize
$
\bullet$\themnotecount}}$ \marginpar{
\raggedright\tiny\em
$\!\!\!\!\!\!\,\bullet$\themnotecount: #1} }

\newcounter{mnotecount}[section]

\renewcommand{\themnotecount}{\thesection.\arabic{mnotecount}}

\newtheorem{theorem}{\sc  Theorem\rm}[section]

\newtheorem{lemma}[theorem]{\sc Lemma\rm}

\newtheorem{remark}[theorem]{\sc Remark\rm}

\newcommand{\ol}[1]{\overline{#1}{}}

\newcommand{\jlcax}[1]{}
\newcommand{\eean}{\nonumber\end{eqnarray}}

\newcommand{\kk}[1]{}

\newcommand{\beq}{\begin{equation}}

\newcommand{\FS}       
                  {F}

\newcommand{\HS} 
       {H_{\mbox{\scriptsize volume}}}

{\ptc{this should be removed in the oberwolfach version}}%

\newcommand{\eeal}[1]{\label{#1}\end{eqnarray}}
\newcommand{\bed}{\begin{deqarr}}
\newcommand{\eed}{\end{deqarr}}
\newcommand{\bedl}[1]{\begin{deqarr}\label{#1}}
\newcommand{\eedl}[2]{\arrlabel{#1}\label{#2}\end{deqarr}}

\newcommand{\bel}[1]{\begin{equation}\label{#1}}
\newcommand{\bea}{\begin{eqnarray}}
\newcommand{\bean}{\begin{eqnarray}\nonumber}
\newcommand{\beal}[1]{\begin{eqnarray}\label{#1}}
\newcommand{\eea}{\end{eqnarray}}

\newcommand{\qed}{\hfill $\Box$ \medskip}

\newcommand{\be}{\begin{equation}}
\newcommand{\eeq}{\end{equation}}
\newcommand{\ee}{\end{equation}}
\newcommand{\beqa}{\begin{eqnarray}}
\newcommand{\eeqa}{\end{eqnarray}}
\newcommand{\beqan}{\begin{eqnarray*}}
\newcommand{\eeqan}{\end{eqnarray*}}
\newcommand{\ba}{\begin{array}}
\newcommand{\ea}{\end{array}}

\newcommand{\mcM}{{\mycal M}}

\newcommand{\scri}{{\mycal I}}%

\newcommand{\warn}[1]
{\protect{\stepcounter{mnotecount}}$^{\mbox{\footnotesize
$
\bullet$\themnotecount}}$ \marginpar{
\raggedright\tiny\em
$\!\!\!\!\!\!\,\bullet$\themnotecount: {\bf Warning:} #1} }

\newcommand{\eq}[1]{(\ref{#1})}

\newcommand{\ptc}[1]{\mnote{{\bf ptc:}#1}}

\newcommand{\mcL}{{\mycal L}}

\newcommand{\beqar}{\begin{deqarr}}
\newcommand{\eeqar}{\end{deqarr}}

\newcommand{\beaa}{\begin{eqnarray*}}
\newcommand{\eeaa}{\end{eqnarray*}}

\DeclareFontFamily{OT1}{rsfs}{}
\DeclareFontShape{OT1}{rsfs}{m}{n}{ <-7> rsfs5 <7-10> rsfs7 <10-> rsfs10}{}
\DeclareMathAlphabet{\mycal}{OT1}{rsfs}{m}{n}

{\catcode `\@=11 \global\let\AddToReset=\@addtoreset}
\AddToReset{equation}{section}

{\catcode `\@=11 \global\let\AddToReset=\@addtoreset}
\AddToReset{figure}{section}

{\catcode `\@=11 \global\let\AddToReset=\@addtoreset}
\AddToReset{table}{section}

\begin{document}

\title{Killing Initial Data on spacelike conformal boundaries%
\thanks{Preprint UWThPh-2013-8. }
}
\author{Tim-Torben Paetz%
\thanks{E-mail:  Tim-Torben.Paetz@univie.ac.at}  \vspace{0.5em}\\  \textit{Gravitational Physics, University of Vienna}  \\ \textit{Boltzmanngasse 5, 1090 Vienna, Austria }}

\maketitle

\vspace{-0.2em}

\begin{abstract}
We analyze Killing Initial Data on Cauchy surfaces in conformally rescaled vacuum space-times satisfying Friedrich's conformal field equations.
As an application, we derive the KID equations on a spacelike~$\scri^-$.
\end{abstract}

\noindent
\hspace{2.1em} PACs number: 04.20.Ex, 04.20.Ha

\tableofcontents

\section{Introduction}

Symmetries are of utmost importance in physics, and so is the construction of space-times $(\tilde{\mcM\enspace}\hspace{-0.5em},\tilde g)$ satisfying Einstein's field equations in general relativity which possess $k$-parameter groups of isometries, $1\leq k\leq 10$ when $\mathrm{dim}\tilde{\mcM\enspace}\hspace{-0.5em}=4$, generated by so-called
Killing vector fields. Indeed, such space-times can be systematically constructed in terms of an initial value problem when the usual constraint equations, which are required to be fulfilled by appropriately prescribed initial data, are supplemented by certain additional equations, the Killing Initial Data (KID) equations.

The KID equations have been derived on spacelike as well as characteristic initial surfaces (cf.\ \cite{beig,CP1} and references therein).
In \cite{ttp2} the same issue was analyzed for characteristic surfaces in conformally rescaled vacuum space-times satisfying Friedrich's conformal field equations.
In particular, for vanishing cosmological constant, the KID equations on a light-cone with vertex at past timelike infinity have been derived there.
The aim of this work is to carry out the corresponding analysis  on spacelike hypersurfaces in conformally rescaled vacuum space-times. As a special case we shall derive the KID equations
on $\scri^-$ supposing that the cosmological constant is positive so that $\scri^-$ is a spacelike hypersurface.

In Section~\ref{sec_setting} we recall the conformal field equations, discuss their gauge freedom and derive the constraint equations induced on $\scri^-$.
Well-posedness of the Cauchy problem for the conformal field equations  with data on $\scri^-$ was shown in \cite{F_lambda}, we shall provide an alternative proof
based on results proved in Appendix~\ref{app_alternative} by using a system of wave equations.

The ``unphysical Killing equations'', introduced in \cite{ttp2} replace, and are in fact equivalent to,
the original-space-time Killing equations in the unphysical space-time.
Employing results in \cite{ttp2} we derive in Section~\ref{sec_KIDs} necessary-and-sufficient conditions on a spacelike hypersurface in a space-time
satisfying the conformal field equations which guarantee  existence of a vector field fulfilling these equations  (cf.\ Theorem~\ref{KID_eqns_main_scri2}).
Similar to the proceeding in \cite{CP1,ttp2} we first derive an intermediate result, Theorem~\ref{KID_eqns_main}, with  a couple of additional hypotheses,
which then are shown to be automatically satisfied.

In Section~\ref{sec_scri} we apply Theorem~\ref{KID_eqns_main_scri2} to the special case where the spacelike hypersurface is $\scri^-$.
We shall see that some of the KID equations determine a set of candidate fields on $\scri^-$. Whether or not these fields extend to vector fields satisfying the
unphysical Killing equations depends on the remaining ``reduced KID equations". As for a light-cone with vertex at past timelike infinity it turns out that
the  KID equations adopt at infinity a significantly simpler form as compared to ``ordinary'' Cauchy surfaces (cf.\ Theorem~\ref{KID_eqns_main_hyp_infinity}).

\section{Setting}
\label{sec_setting}

\subsection{Conformal field equations}

In $3+1$ dimensions Friedrich's \textit{metric conformal field equations (MCFE)} (cf.~\cite{F3})%
\footnote{It is indicated in \cite{ttp2} that things are considerably different in higher dimensions, which is why we restrict attention to 4 dimensions from
the outset.}
\begin{eqnarray}
 && \nabla_{\rho} d_{\mu\nu\sigma}{}^{\rho} =0\;,
 \label{conf1}
\\
 && \nabla_{\mu} L_{\nu\sigma} - \nabla_{\nu}L_{\mu\sigma} =\nabla_{\rho}\Theta \, d_{\nu\mu\sigma}{}^{\rho}\;,
 \label{conf2}
\\
 && \nabla_{\mu}\nabla_{\nu}\Theta = -\Theta L_{\mu\nu} + s g_{\mu\nu}\;,
 \label{conf3}
\\
 && \nabla_{\mu} s = -L_{\mu\nu}\nabla^{\nu}\Theta\;,
 \label{conf4}
\\
 && 2\Theta s - \nabla_{\mu}\Theta\nabla^{\mu}\Theta = \lambda/3 \;,
 \label{conf5}
\\
 && R_{\mu\nu\sigma}{}^{\kappa}[ g] = \Theta d_{\mu\nu\sigma}{}^{\kappa} + 2(g_{\sigma[\mu} L_{\nu]}{}^{\kappa}  - \delta_{[\mu}{}^{\kappa}L_{\nu]\sigma} )
 \label{conf6}
\end{eqnarray}
form a closed system of equations for the unknowns
$g_{\mu\nu}$, $\Theta$, $s$, $ L_{\mu\nu}$ and $d_{\mu\nu\sigma}{}^{\rho}$.
The tensor field $L_{\mu\nu}$ denotes the Schouten tensor,
\begin{equation}
  L_{\mu\nu} \,=\, \frac{1}{2}R_{\mu\nu} - \frac{1}{12} R g_{\mu\nu}
 \;,
\label{dfn_L}
\end{equation}
while
\begin{equation}
d_{\mu\nu\sigma}{}^{\rho}\,=\, \Theta^{-1}C_{\mu\nu\sigma}{}^{\rho}
\label{dfn_d}
\end{equation}
 is a rescaling of the conformal Weyl tensor $C_{\mu\nu\sigma}{}^{\rho}$.
The function $s$ is defined as
\begin{equation}
s \,=\, \frac{1}{4}\Box_g\Theta + \frac{1}{24} R\Theta
\;.
\label{dfn_s}
\end{equation}
%

Friedrich has shown that the MCFE are equivalent to Einstein's vacuum field equations with cosmological constant $\lambda$ in regions where the conformal factor $\Theta$, relating the ``unphysical'' metric $g=\Theta^2g_{\mathrm{phys}}$ with the physical metric $g_{\mathrm{phys}}$,  is positive. Their advantage lies in the property that they remain regular even where $\Theta$ vanishes.

The system \eq{conf1}-\eq{conf6} treats $s$, $L_{\mu\nu}$ and $d_{\mu\nu\sigma}{}^{\rho}$ as independent of $g_{\mu\nu}$ and $\Theta$.
However, once a solution of the MCFE has been given these fields are related to $g_{\mu\nu}$ and $\Theta$ via \eq{dfn_L}-\eq{dfn_s}.
A solution of the MCFE is thus completely determined by the pair
 $(g_{\mu\nu},\Theta)$.

\subsection{Gauge freedom}
\label{sec_gauge_freedom}

\subsubsection{Conformal factor}

Let  $(g_{\mu\nu}, \Theta, s, L_{\mu\nu}, d_{\mu\nu\sigma}{}^{\rho})$ be some smooth solution of the MCFE.%
\footnote{For convenience we restrict attention throughout to the smooth case, though similar results can be obtained assuming finite differentiability.}
From $g_{\mu\nu}$ we compute $R$. Let us then conformally rescale the metric, $g\mapsto \phi^2g$, for some positive function $\phi> 0$.
The Ricci scalars $R$ and $R^*$ of $g$ and $\phi^2g$, respectively, are related via (set $\Box_g:= g^{\mu\nu}\nabla_{\mu}\nabla_{\nu}$)
\begin{equation}
 \phi R - \phi^3 R^* = 6\Box_g\phi
 \;.
  \label{R_R*}
\end{equation}
Now, let us \textit{prescribe} $R^*$ and read \eq{R_R*} as an equation for $\phi$.
When dealing with a Cauchy problem with data on some spacelike hypersurface $\mathcal{ H}$ (including $\scri^-$ for $\lambda>0$) we are free to prescribe functions $\phi|_{\mathcal{H}}=:\mathring \phi>0$ and  $\partial_0\phi|_{\mathcal{H}}=:\mathring \psi$ on $\mathcal{H}$.%
\footnote{The positivity-assumption on $\mathring \phi$ makes sure that the solution of \eq{R_R*} is positive sufficiently close to $\mathcal{H}$  and thereby that the new conformal factor $\Theta^*$  is positive as well (in the $\scri^-$-case just off the initial surface).
}
%
Throughout $x^0\equiv t$ denotes a time-coordinate so that $\partial_0$ is transverse to $\mathcal{H}$.
According to standard results  there exists a unique solution $\phi>0$  in some neighborhood of $\mathcal{H}$ which induces the above data on $\mathcal{H}$.
The MCFE are conformally covariant, meaning that the conformally rescaled fields
\begin{eqnarray}
 g^* &=& \phi^2 g\;,
 \label{conformal_covariance1}
 \\
 \Theta^* &=& \phi \,\Theta\;,
 \\
 s^* &=& \frac{1}{4}\Box_{g^*}\Theta^* + \frac{1}{24}R^*\Theta^* \;,
   \\
   L^*_{\mu\nu} &=& \frac{1}{2} R^*_{\mu\nu}[g^*] - \frac{1}{12}R^* g^*_{\mu\nu} \;,
   \\
   d^*_{\mu\nu\sigma}{}^{\rho} &=& \phi^{-1} d_{\mu\nu\sigma}{}^{\rho}\;,
    \label{conformal_covariance5}
  \end{eqnarray}
provide another solution of the MCFE, now with Ricci scalar $R^*$, which represents the same physical solution:
If the conformal factor $\Theta$ is treated as an unknown, determined by the MCFE,
the unphysical Ricci scalar $R$ can be arranged to adopt any preassigned form,
 it represents a \textit{conformal gauge source function}. 

There remains the gauge freedom to prescribe the functions $\mathring \phi$ and $\mathring \psi$ on $\mathcal{H}$.
On an ordinary hypersurface, where $\Theta$ has no zeros, this freedom can be used to prescribe 
$\Theta|_{\mathcal{H}}$ and $\partial_0\Theta|_{\mathcal{H}}$.
A main object of this work is to treat the case  $\mathcal{H}=\scri^-$, where, by definition, $\Theta =0$ (and $\mathrm{d}\Theta\ne 0$).
We shall show that in this situation the gauge freedom allows one to prescribe the function $s$  on $\scri^-$ and
to make  conformal rescalings of the induced metric on $\scri^-$.

To see this  we consider a smooth solution 
of the MCFE
to the future of $\scri^-$.
Now \eq{conf5} and  $\mathrm{d}\Theta|_{\scri^-}\ne 0$ enforce $\ol g{}^{00}< 0$ (hence, as is well known, $
\scri^-$ must be spacelike when $\lambda>0$).
Due to \eq{conf5}, the function $s$ can be written away from $\scri^-$ as
\[ s= \frac{1}{2}\Theta^{-1} \nabla_{\mu}\Theta\nabla^{\mu}\Theta + \frac{1}{6}\Theta^{-1}\lambda\;, \]
and the right-hand side is smoothly extendable at $\scri^-$.
A conformal rescaling
\begin{eqnarray}
\Theta \mapsto \Theta^*:=\phi\,\Theta\;, \quad  g_{\mu\nu} \mapsto g^*_{\mu\nu}:=\phi^2 g_{\mu\nu}\;, \quad \phi>0\;,
\label{rescaling}
\end{eqnarray}
maps the function $s$ to
\begin{equation}
 s^* = \phi^{-1}\Big( \frac{1}{2}\Theta \phi^{-2}\nabla^{\mu}\phi \nabla_{\mu}\phi + \phi^{-1} \nabla^{\mu}\Theta \nabla_{\mu}\phi + s\Big)\;.
\label{trafo_s_s'}
\end{equation}
The trace of this equation on $\scri^-$ is
\begin{eqnarray}
 \overline{\nabla^{\mu}\Theta\nabla_{\mu} \phi +  \phi\,  s - \phi^2 s^*}=0
  \;,
\end{eqnarray}
or, in coordinates adapted to $\scri^-$, i.e.\ for which $\scri^- = \{ x^0\equiv t =0\}$ locally,
\begin{eqnarray}
\ol{g^{0\mu} \nabla_{0}\Theta\nabla_{\mu} \phi +  \phi\,  s - \phi^2 s^*}=0
  \;.
\label{eqn_phi}
\end{eqnarray}
Here and henceforth we use overlining to denote restriction to the initial surface.
Let us prescribe $\overline s^*$ .
We choose any $\mathring \phi>0$ to conformally rescale the induced metric on $\scri^-$.
Then we solve \eq{eqn_phi} for $\mathring\psi\equiv \ol{\nabla_0 \phi}$ (recall that$\ol {\nabla_0 \Theta}$ and    $\ol g{}^{00}$ are not allowed to have zeros on $\scri^-$).
We take the so-obtained functions $\mathring\phi>0$ and $\mathring \psi$  as initial data for \eq{R_R*}.

By way of summary, the conformal covariance of the MCFE  comprises a gauge freedom due to which
the functions $R$ and $s|_{\scri^-}$ can  be regarded as gauge source functions, and due to which only the conformal class of
 the induced metric on $\scri^-$ matters.

\subsubsection{Coordinates}

It is well-known (cf.\ e.g.\ \cite{c_lecture}) that the freedom to choose coordinates near
 a spacelike hypersurface $\mathcal{H}=\{x^0=0\}$ with induced Riemannian metric $h_{ij}$ can be employed to prescribe
\begin{equation}
 \ol g^{00}<0   \quad \text{and}\quad   \ol g^{0i}
\;.
\end{equation}
Equivalently, one may prescribe
\begin{equation}
 \ol g_{00}   \quad \text{and}\quad   \ol g_{0i} \quad \text{such that} \quad \ol g_{00} -  \ol h^{ij} \ol g_{0i}\ol g_{0j} <0
\;.
\end{equation}
The remaining  freedom to choose coordinates off
the initial surface is comprised in the \textit{$\hat g$-generalized wave-map gauge condition}
\begin{equation}
H^{\sigma}=0
\end{equation}
with
\begin{equation}
  H^{\sigma} := g^{\alpha\beta}(\Gamma^ {\sigma}_{\alpha\beta} -\hat\Gamma^{\sigma}_{\alpha\beta} ) -  W^{\sigma}
\end{equation}
being the \textit{generalized wave-gauge vector}.
Here $\hat g_{\mu\nu}$ denotes some \textit{target metric}, $\hat\Gamma^{\sigma}_{\alpha\beta}$ are the Christoffel symbols of  $\hat g_{\mu\nu}$.
More precisely, the gauge freedom is captured by the vector field
\begin{equation*}
 W^{\sigma} = W^{\sigma}(x^{\mu}, g_{\mu\nu},  s,\Theta, L_{\mu\nu}, d_{\mu\nu\sigma}{}^{\rho},\hat g_{\mu\nu})
\end{equation*}
which can be  arbitrarily prescribed. In fact, within our setup, it can be allowed to depend upon the coordinates, and possibly upon $g_{\mu\nu}$ as well as all other fields which appear in the MCFE,
but not upon derivatives thereof.

\subsubsection{Realization of the gauge scheme}
\label{sec_realization}

Given some smooth solution of the MCFE and a new choice of gauge functions  $R$, $\ol s$, $W^{\sigma}$, $\ol g_{0\mu}$, as well as a  conformal factor $\Omega>0$ by which one wants to rescale the induced metric $\ol g_{ij}$,
a transformation into the new gauge is realized as follows:

In the first step we set $\mathring \phi := \Omega$ and solve \eq{eqn_phi} for
$\mathring\psi\equiv \ol{\nabla_0 \phi}$, which gives us the relevant initial data for \eq{R_R*} which we then solve.
This way $\ol s$ and $R$ take their desired values, and a new representative $\Omega^2 \ol g_{ij}$ of the conformal class  of the induced metric
on $\scri^-$ is selected.
Then the coordinates are transformed in such a way that the metric takes the prescribed values for
$\ol g_{0\mu}$ on $\scri^-$.
Finally we just need to solve another wave equation to obtain $H^{\sigma}=0$ for the given vector field
$W^{\sigma}$.


\subsection{Constraint equations in the $(R=0, W^{\lambda}=0, \ol s=0,  \ol g_{00}=-1, \ol g_{0i}=0, \hat g_{\mu\nu} = \ol g_{\mu\nu})$-wave map gauge}
\label{subsec_gauge}

In the following we aim to derive the constraint equations for
 the fields  $g_{\mu\nu}$, $\Theta$, $s$, $L_{\mu\nu}$, $d_{\mu\nu\sigma}{}^{\rho}$ on $\scri^-$ as well as their transverse derivatives
induced by the MCFE on a spacelike $\scri^-$ in adapted coordinates $(x^0=t, x^i)$ with $\scri^-=\{t=0\}$. The surface $\scri^-$ is characterized by
\begin{equation}
\ol \Theta \,=\,0 \quad \text{and} \quad  \ol{\mathrm{d}\Theta} \ne 0
\;.
\end{equation}
Note that for $\scri^-$ to be spacelike a positive cosmological constant $\lambda>0$ is required.
The constraint equations will be relevant for the derivation of the KID equations in Section~\ref{sec_scri}.

To simplify computations we make the specific gauge choice
\begin{eqnarray}
 R=0\;, \quad \ol s=0\;, \quad \ol g_{00}=-1\;, \quad \ol g_{0i}=0\;, \quad  W^{\sigma}=0\;, \quad \hat g_{\mu\nu} = \ol g_{\mu\nu}
\;.
\label{gauge_conditions_compact}
\end{eqnarray}
(Note that the target metric is taken to be $\ol g_{\mu\nu}$ for all $t$.)
We shall show that appropriate data to solve the constraint equations are $\ol g_{ij}$ and $\ol d_{0i0j}$, where the latter  field  needs to satisfy a vector and a scalar constraint equation.

Let us start with a list of all the Christoffel symbols in adapted coordinates
\begin{eqnarray}
 &\ol \Gamma^k_{ij} = \tilde\Gamma^k_{ij}\;, \quad  \ol \Gamma^0_{ij} =  \frac{1}{2}\ol{\partial_0 g_{ij}}\;, \quad  \ol \Gamma^0_{0i} = 0\;,&
\label{christoffel1}
\\
& \ol \Gamma^0_{00} =-\frac{1}{2}\ol{\partial_0 g_{00}}\;, \quad  \ol \Gamma^k_{00} =  \ol g^{kl}\ol{\partial_0 g_{0l}}\;, \quad  \ol \Gamma^k_{0i} = \frac{1}{2}\ol g^{kl}\ol{\partial_0 g_{il}}\;,&
\label{christoffel2}
\end{eqnarray}
where the $ \tilde\Gamma^k_{ij}$'s denote the Christoffel symbols of the Riemannian metric $\tilde g= \ol g_{ij}\mathrm{d}x^i\mathrm{d}x^j$.
Throughout we shall use $\tilde . $ to denote fields such as the  Riemann tensor,  the Levi-Civita connection etc.\ associated to $\tilde g$.

Evaluation  of \eq{conf5} on $\scri^-$ gives
\begin{equation}
 \ol{\nabla_0\Theta} = \sqrt{\lambda / 3}
\;.
\end{equation}
The $(\mu\nu)=(00)$-component of \eq{conf3} implies
\begin{equation}
\ol{\nabla_0\nabla_0\Theta} \,=\,0
\;,
\end{equation}
while the $(\mu\nu)=(ij)$-components of \eq{conf3} yield
\begin{equation}
 0 = \ol{\nabla_i\nabla_j\Theta} = - \ol \Gamma^0_{ij} \ol{\nabla_0\Theta} = - \sqrt{\frac{\lambda}{12}}\, \ol{\partial_0 g_{ij}}
\;.
\end{equation}
We compute the $(\mu\nu\sigma\kappa)=(ikjk)$-components of \eq{conf6},
\begin{eqnarray*}
\ol R_{ikj}{}^{k}= \ol L_{ij}  + g_{ij}\ol g^{kl}\ol L_{kl}
\;,
\end{eqnarray*}
where
\begin{eqnarray*}
 \ol R_{ikj}{}^{k} \,=\, \partial_k \ol \Gamma^k_{ij} - \partial_i \ol\Gamma^k_{jk} + \ol \Gamma^{\alpha}_{ij}\ol \Gamma^k_{\alpha k} - \ol \Gamma^{\alpha}_{ik}\ol \Gamma^k_{j\alpha}
\,=\, \tilde   R_{ikj}{}^{k} \,=\, \tilde R_{ij}
\;.
\end{eqnarray*}
Hence
\begin{equation}
 \ol L_{ij} \,=\, \tilde R_{ij} -  \frac{1}{4}\ol g_{ij}\tilde R \,=\, \tilde L_{ij}
\;,
\label{constraint_Lij}
\end{equation}
where $\tilde L_{ij}$ is the \textit{Schouten tensor} of $\tilde g$.
The gauge conditions \eq{gauge_conditions_compact} imply
\begin{equation}
 0 =\frac{1}{6}\ol R = \ol g^{\mu\nu}\ol L_{\mu\nu} = \ol g^{ij} \ol L_{ij} - \ol L_{00} =\frac{1}{4}\tilde R -\ol L_{00}
\;.
\end{equation}
From the $\mu=i$-component of \eq{conf4} we deduce
\begin{equation}
 \ol L_{0i} =0
\;.
\end{equation}
Next, we employ the wave-map gauge condition to obtain
\begin{eqnarray*}
 0 &=& \ol H^k = g^{\alpha\beta}(\Gamma^k_{\alpha\beta}- \hat\Gamma^k_{\alpha\beta}) = -\ol \Gamma^k_{00} = -\ol g^{kl}\ol{\partial_0 g_{0l}}
\;,
\\
 0 &=& \ol H^0 = g^{\alpha\beta}(\Gamma^0_{\alpha\beta}- \hat\Gamma^0_{\alpha\beta}) =-\ol \Gamma^0_{00} = \frac{1}{2}\ol{\partial_0 g_{00}}
\;.
\end{eqnarray*}
Altogether we have found that
\begin{equation}
 \ol{\partial_0 g_{\mu\nu}} \,=\, 0 \;.
\end{equation}
Thus \eq{christoffel1}-\eq{christoffel2} simplify to
\begin{eqnarray}
 \ol \Gamma^k_{ij} = \tilde\Gamma^k_{ij}\;, \quad  \ol \Gamma^0_{ij} =    \ol \Gamma^0_{0i} =\ol \Gamma^0_{00} =  \ol \Gamma^k_{00} =  \ol \Gamma^k_{0i} = 0\;.
\label{christoffel}
\end{eqnarray}

We have
\begin{eqnarray*}
 \ol R_{ij} &\equiv&  \ol{\partial_{\mu}\Gamma^{\mu}_{ij}} - \ol{\partial_i \Gamma^{\mu}_{j\mu}} + \ol \Gamma^{\alpha}_{ij}\ol \Gamma^{\mu}_{\alpha\mu} - \ol \Gamma^{\alpha}_{i\mu}\ol \Gamma^{\mu}_{j\alpha}
\\
&=& \tilde R_{ij} +  \ol{\partial_{0}\Gamma^{0}_{ij}}
=  \tilde R_{ij} + \frac{1}{2}\ol{\partial_0\partial_0 g_{ij}}
\;.
\end{eqnarray*}
Hence
\begin{eqnarray}
 \ol{\partial_0\partial_0 g_{ij}} &=& 4\ol L_{ij} - 2\tilde R_{ij} \,=\,  2\tilde R_{ij} -  g_{ij}\tilde R
\;.
\end{eqnarray}
If we evaluate the $\mu=0$-component of \eq{conf4} on $\scri^-$ we are led to,
\begin{equation}
 \ol{\nabla_0s} = \ol L_{00}\ol{\nabla_0 \Theta} =  \sqrt{\frac{\lambda}{48}} \,\tilde R
\;.
\end{equation}
The $(\mu\nu\sigma)=(0i0)$-components of \eq{conf2} yield
\begin{equation}
 \ol{\nabla_0 L_{0i}} = \nabla_i \ol L_{00}  = \frac{1}{4}\tilde\nabla_i\tilde R
\;.
\end{equation}
Moreover, for  $(\mu\nu\sigma)=(jki)$ we obtain
\begin{eqnarray}
 \ol d_{0ijk}  &=&    \sqrt{\frac{12}{\lambda}} \,\tilde\nabla_{[k}\ol L_{j] i}
\,=\,
 \sqrt{\frac{3}{\lambda}} \tilde C_{ijk}
\;,
\end{eqnarray}
where $\tilde C_{ijk}$ is the \textit{Cotton tensor} of $\tilde g$.
For  $(\mu\nu\sigma)=(0ji)$ we find
\begin{eqnarray}
\ol{\nabla_{0}L_{ ij}}  &=&  -\sqrt{\lambda/3} \, \ol d_{0i 0j}
\;.
\end{eqnarray}
The gauge condition $R=0$ together with the tracelessness of the rescaled Weyl tensor then imply
\begin{equation}
0 \,=\, \ol g^{\mu\nu} \ol{\nabla_0 L_{\mu\nu}} \,=\, \ol g^{ij} \ol{\nabla_0 L_{ij}}  -  \ol{ \nabla_0 L_{00}} \,=\, -  \ol{ \nabla_0 L_{00}}
\;.
\end{equation}
Via the second Bianchi identity  the $(\mu\nu\sigma)=(0ij)$-components of \eq{conf1} become
\begin{eqnarray}
 \ol{\nabla_{0} d_{0i0j}}&=&- \tilde \nabla^{k}\ol  d_{0ijk}
 \,=\, -\sqrt{\frac{3}{\lambda}}\tilde \nabla^{k} \tilde C_{ijk}
 \,=\,\sqrt{\frac{3}{\lambda}}\tilde B_{ij}
\;,
\end{eqnarray}
where $\tilde B_{ij}$ denotes the \textit{Bach tensor} of $\tilde g$.
The  $(\mu\nu\sigma)=(kji)$-components give
\begin{equation}
\ol{\nabla_{0} d_{0ijk}} \,= \, - \tilde \nabla^{l} \ol d_{jkil}  \,=\,  2\tilde \nabla_{[j}\ol d_{k]0i0} - 2g_{i[j}\tilde \nabla^{l}\ol d_{k]0l0}
\;.
\end{equation}
Here  we used that due to the algebraic symmetries of the rescaled Weyl tensor
\begin{eqnarray}
\ol d_{ijkl} &=& 2\ol  g^{mn}(\ol g_{k[i}\ol d_{j]mln} - \ol  g_{l[i}\ol d_{j]mkn} - \ol g_{k[i}g_{j]l}\ol g^{pq} \ol d_{pmqn})
\nonumber
\\
&=&  2(\ol g_{k[i}\ol d_{j]0l0} - \ol  g_{l[i}\ol d_{j]0k0} )
\;.
\label{spatial_tr}
\end{eqnarray}
The $(\mu\nu\sigma)=(0i0)$-components of \eq{conf1} imply a \emph{vector constraint} for $\ol d_{0i0j}$,
\begin{equation}
  \tilde\nabla^{j} \ol d_{0i0j} \,=\,0
\;.
\end{equation}
(A \emph{``scalar constraint''}, which has already been used in the derivation of the constraint equations,
is simply given by the tracelessness-requirement on the rescaled Weyl tensor,
\begin{equation}
 \ol g^{ij}\ol d_{0i0j} \,=\,   \ol g^{\mu\nu}\ol d_{0\mu 0\nu} \,=\,   0
\;.)
\end{equation}

To sum it up, we have the following analogue of a result of Friedrich~\cite{F_lambda}:
The free data can be identified with a Riemannian metric $h_{ij}:=\ol g_{ij}$ and a symmetric tensor field $D_{ij}:= \ol d_{0i0j}$
on $\scri^-$ satisfying
\begin{equation}
 h^{ij} D_{ij}=0 \quad \text{and} \quad \tilde\nabla^j D_{ij}=0
\label{constr1}
\end{equation}
(that these are indeed the free data follows e.g.\ from the considerations in Appendix~\ref{app_alternative}).
Then the MCFE enforce on $\scri^-$ in the  $(R=0, \ol s=0,  \ol g_{00}=-1, \ol g_{0i}=0, \hat g_{\mu\nu} = \ol g_{\mu\nu})$-wave-map gauge,
\begin{eqnarray}
 &\ol g_{00} =-1\;, \quad \ol g_{0i}=0\;, \quad \ol g_{ij} = h_{ij}\;, \quad \ol{\partial_0 g_{\mu\nu}}=0\;,
\label{constr2}
&
\\
& \ol \Theta =0\;, \quad \ol{\partial_0 \Theta} = \sqrt{\frac{\lambda}{3}}\;,
\label{constr3}
&
\\
&\ol s=0\;, \quad \ol{\partial_0 s} =  \sqrt{\frac{\lambda}{48}} \,\tilde R\;,
\label{constr4}
&
\\
&\ol L_{ij} = \tilde L_{ij} \;, \quad \ol L_{0i}=0\;, \quad \ol L_{00} = \frac{1}{4}\tilde R \;,
\label{constr5}
&
\\
& \ol{\partial_0 L_{ij}} =  -\sqrt{\frac{\lambda}{3}} \, D_{ij}\;, \quad
  \ol{\partial_0 L_{0i}} =\frac{1}{4}\tilde\nabla_i\tilde R  \;, \quad  \ol{\partial_0 L_{00}} =0\;,
\label{constr6}
&
\\
&  \ol d_{0i0j} = D_{ij}\;,\quad \ol d_{0ijk}  =  \sqrt{\frac{3}{\lambda}} \tilde C_{ijk}\;,
\label{constr7}
&
\\
& \ol{\partial_{0} d_{0i0j}}= \sqrt{\frac{3}{\lambda}}\tilde B_{ij}\;, \quad
\ol{\partial_{0} d_{0ijk}} = 2\tilde \nabla_{[j}D_{k]i}\;.
\label{constr8}
 &
\end{eqnarray}
Note that due to \eq{christoffel} the actions of $\nabla_0$ and $\partial_0$, as well as $\nabla_i$ and $\tilde\nabla_i$, respectively, coincide on $\scri^-$,
so we can use them interchangeably.

We have seen in Section~\ref{sec_gauge_freedom} (cf.\ also \cite{F_lambda}) that there remains a gauge freedom to conformally rescale the induced metric on $\scri^-$. Due to this freedom 
the pairs $(h_{ij},D_{ij})$ and $(\Omega^2 h_{ij}, \Omega^{-1} D_{ij})$,
with $\Omega$ some positive function, generate the same physical space-times.
With regard to the constraint equations we note that $\Omega^{-1}D_{ij}$ is trace- and divergence-free w.r.t.\ $\Omega^2 h_{ij}$ whenever $D_{ij}$ is w.r.t.~$h_{ij}$.

In the following we shall write $[h_{ij},D_{ij}]$ if this gauge freedom is left unspecified and if we merely want to refer to the
conformal classes of $h_{ij}$ and~$D_{ij}$.

\subsection{Well-posedness of the Cauchy problem on a spacelike~$\scri^-$}

In \cite{ttp} a system of conformal wave equations (CWE) has been derived from the MCFE.
In Appendix~\ref{app_alternative} it is shown that a solution of the CWE, equations \eq{cwe1}-\eq{cwe5},  is a solution of the MCFE if and only if the constraint equations \eq{constr1}-\eq{constr8} are satisfied.
Using standard well-posedness results about wave equations we thereby recover a result due to Friedrich \cite{F_lambda} who proved well-posedness of the Cauchy problem
on $\scri^-$ (Friedrich used a representation of the MCFE as a symmetric hyperbolic system, in some situations, however, it might be advantageous
to deal with a system of wave equations instead \cite{kreiss}). 
We restrict attention to the smooth case (for a version with finite differentiability see \cite{F_lambda}):

\begin{theorem}
\label{thm_well-posedness}
Let $\mathcal{H}$ be a 3-dimensional smooth manifold. Let $h_{ij}$ be a smooth Riemannian metric and let $D_{ij}$ be a smooth symmetric, trace- and divergence-free tensor field on $\mathcal H$.
Moreover, assume a positive cosmological constant $\lambda>0$.
Then there exists an  (up to isometries)
unique smooth space-time $(\mcM, g,\Theta)$ with the following properties:
\begin{enumerate}
\item[(i)]  $(\mcM, g,\Theta)$ satisfies the MCFE \eq{conf1}-\eq{conf6},
\item[(ii)] $\Theta|_\mathcal{H}=0$ and  $\mathrm{d}\Theta|_\mathcal{H}\ne 0$, i.e.\ $\mathcal{H}=\scri^-$ (and $\Theta$ has no zeros away from and sufficiently close to $\mathcal{H}$),
\item[(iii)] $g_{ij}|_{\mathcal{H}} = h_{ij}$, $d_{0i0j}|_{\mathcal{H}}=D_{ij}$.
\end{enumerate}
The isometry class of the space-time does not change if the initial data are replaced by $(\hat h_{ij}, \hat D_{ij})$
with $[\hat h_{ij},\hat D_{ij}] =[ h_{ij},D_{ij}]$.
\end{theorem}

\begin{remark}
{\rm
De Sitter space-time is obtained for $\mathcal{H}=S^3$, $h_{ij}=s_{ij}$ and $D_{ij}=0$, where $s=s_{ij}\mathrm{d}x^i\mathrm{d}x^j$ denotes the round sphere metric,
cf.\ Section~\ref{max_symm}
}
\end{remark}

\section{KID equations}
\label{sec_KIDs}

\subsection{Unphysical Killing equations}

In \cite{ttp2} it is shown that the appropriate substitute for the Killing equation in the unphysical, conformally rescaled space-time is provided by the
\textit{unphysical Killing equations}
\begin{equation}
 \nabla_{(\mu}X_{\nu)} =\frac{1}{4}\nabla^{\sigma} X_{\sigma}\,g_{\mu\nu}
   \quad \& \quad X^{\sigma}\nabla_{\sigma}\Theta =\frac{1}{4}\Theta \nabla_{\sigma} X^{\sigma}
 \;.
\label{2_conditions}
\end{equation}
A vector field $X_{\mathrm{phys}}$ is a Killing field in the physical space-time $(\mcM_{\mathrm{phys}},g_{\mathrm{phys}})$ if and only if
its push-forward
 $X:= \phi_*X_{\mathrm{phys}}$ satisfies \eq{2_conditions} in the unphysical space-time $(\phi(\mcM_{\mathrm{phys}})\subset\mcM,g =\phi(g_{\mathrm{phys}})=\Theta^2  g_{\mathrm{phys}} )$, where $\phi$ defines the conformal rescaling.
The unphysical Killing equations remain regular even where the conformal factor $\Theta$ vanishes.

In what follows
 we shall derive necessary-and-sufficient conditions on a spacelike initial surface which guarantee the existence of a vector field $X$
which satisfies the unphysical Killing equations.

\subsection{KID equations on a Cauchy surface}

Necessary conditions on a vector field $X$ to satisfy the unphysical Killing equations are that
the following wave equations are fulfilled \cite{ttp2},
%
\begin{eqnarray}
\Box_g X_{\mu} + R_{\mu}{}^{\nu}X_{\nu}  + 2\nabla_{\mu}Y &=&0
 \label{wave_X0}
\;,
\\
 \Box_g Y+ \frac{1}{6}X^{\mu}\nabla_{\mu}R+ \frac{1}{3} R Y   &=&0
 \label{wave_Y0}
 \;,
\end{eqnarray}
where we have set
\begin{equation}
 Y:= \frac{1}{4}\nabla_{\sigma}X^{\sigma}
\;.
\end{equation}
It proves fruitful to make the following definitions:
\begin{eqnarray}
 \phi &:=&  X^{\mu}\nabla_{\mu}\Theta - \Theta Y
\;,
\\
\psi &:=& X^{\mu}\nabla_{\mu}s    +sY -\nabla_{\mu}\Theta \nabla^{\mu}Y
\;,
\\
A_{\mu\nu} &:=& 2\nabla_{(\mu}X_{\nu)}  - 2Y g_{\mu\nu}
\;,
\\
B_{\mu\nu} &:=& \mcL_XL_{\mu\nu} +   \nabla_{\mu}\nabla_{\nu}Y
\;.
\end{eqnarray}
All these fields need to vanish whenever $X$ is a solution of \eq{2_conditions} \cite{ttp2}.

The equations \eq{wave_X0} and \eq{wave_Y0} together with the MCFE imply that the following system of wave equations is satisfied by
the fields $\phi$, $\psi$, $A_{\mu\nu}$, $\nabla_{\sigma}A_{\mu\nu}$ and $B_{\mu\nu}$ (cf. \cite{ttp2}):
\begin{eqnarray}
 \Box_gA_{\mu\nu}  &=& 2R_{(\mu}{}^{\kappa}A_{\nu)\kappa} - 2R_{\mu}{}^{\alpha}{}_{\nu}{}^{\beta}A_{\alpha\beta}
  -4B_{\mu\nu}
 \;,
      \label{waveeqn_A}
\\
 \Box_g\phi
 &=& d\psi - \frac{1}{6}R\phi+ A_{\mu\nu} \nabla^{\mu}\nabla^{\nu}\Theta
 \;,
  \label{waveeqn_XnablaTheta}
\\
 \Box_g\psi
  &=& |L|^2 \phi
 + A_{\mu\nu}(\nabla^{\mu}\nabla^{\nu}s - 2\Theta  L_{\kappa}{}^{\mu} L^{\nu\kappa}) + 2\Theta L^{\mu\nu} B_{\mu\nu}
\nonumber
\\
&& + \frac{1}{6}\big(  A_{\mu\nu}\nabla^{\mu}R  \nabla^{\nu}\Theta   - \nabla^{\mu}R \nabla_{\mu}\phi - R\psi   \big)
 \;,
  \label{waveeqn_Xnabla_s}
\\
\Box_g B_{\mu\nu} &\equiv&
 2(g_{\mu\nu}L^{\alpha \beta}-   R_{\mu}{}^{\alpha}{}_{\nu}{}^{\beta})B_{\alpha\beta}
 - 2 R_{(\mu}{}^{\kappa} B_{\nu)\kappa}
+  \frac{2}{3}RB_{\mu\nu}
\nonumber
\\
&&\hspace{-3em}
   + 2L^{\alpha\beta}(\nabla_{\beta}\nabla_{[\alpha}A_{\nu]\mu}  -  \nabla_{\mu}\nabla_{[\alpha}A_{\nu]\beta})
\nonumber
\\
&&\hspace{-3em}
+ (\nabla_{(\mu}A_{|\alpha\beta|}  +  2\nabla_{[\alpha}A_{\beta](\mu})(2\nabla^{\alpha}L_{\nu)}{}^{\beta}
-\frac{1}{12}\delta_{\nu)}{}^{\alpha}\nabla^{\beta}R)
\nonumber
\\
&&\hspace{-3em}
+A^{\alpha\beta} [\nabla_{\alpha}\nabla_{\beta}L_{\mu\nu}
- 2 L_{(\mu}{}^{\kappa} R_{\nu)\alpha\kappa\beta}
+ 2L_{\mu\alpha} R_{\nu\beta}
+ L_{\alpha}{}^{\kappa} (2R_{\mu\beta\nu\kappa}+  R_{\nu\beta\mu\kappa})
\nonumber
\\
&&\hspace{-3em}
-  2g_{\mu\nu} L_{\alpha\kappa}L_{\beta}{}^{\kappa}]
 + |L|^2 A_{\mu\nu}
+  L^{\alpha\beta} R_{\mu\alpha\beta}{}^{\kappa} A_{\nu\kappa}
-  \frac{1}{3}RL_{(\mu}{}^{\kappa}A_{\nu)\kappa}
 \;,
\label{waveeqn_B}
\\
\Box_g\nabla_{\sigma}A_{\mu\nu}  &=&
  2\nabla_{\sigma}(R_{(\mu}{}^{\kappa}A_{\nu)\kappa}
 - R_{\mu}{}^{\alpha}{}_{\nu}{}^{\kappa}A_{\alpha\kappa}  )
   + 2A_{\alpha(\mu}(\nabla_{\nu)}R_{\sigma}{}^{\alpha} - \nabla^{\alpha}R_{\nu)\sigma} )
\nonumber
\\
 &&
  -4R_{\sigma\kappa(\mu}{}^{\alpha}\nabla^{\kappa}A_{\nu)\alpha}
+ R_{\alpha\sigma}\nabla^{\alpha}A_{\mu\nu}
  - 4\nabla_{\sigma}B_{\mu\nu}
      \label{waveeqn_nabla_A}
\;.
\end{eqnarray}
In close analogy to \cite[Theorem 3.4]{ttp2} we immediately obtain the following result:
\begin{theorem}
 \label{KID_eqns_main}
Assume we have been given, in $3+1$ dimensions,  an``unphysical'' space-time $(\mcM, g, \Theta)$, with ($g, \Theta$) a smooth solution of the MCFE \eq{conf1}-\eq{conf6}.
Consider a spacelike hypersurface $\mathcal{H}\subset \mcM$.
 Then there exists a vector field $\hat X$ satisfying the unphysical Killing equations \eq{2_conditions} on $\mathrm{D}^+(\mathcal{H})$
(and thus corresponding to a Killing vector field of the physical space-time)
if and only if there exists a pair $(X,Y)$, $X$ a vector field and $Y$ a function, which fulfills the following equations:
 \begin{enumerate}
  \item[(i)] $\Box_g X_{\mu} + R_{\mu}{}^{\nu}X_{\nu}  + 2\nabla_{\mu}Y= 0$,
  \item[(ii)] $ \Box_g Y+ \frac{1}{6}X^{\mu}\nabla_{\mu}R+ \frac{1}{3} R Y =0$,
   \item[(iii)] $ \ol\phi = 0$ and $ \ol{\partial_0 \phi} = 0$,
  \item[(iv)] $\ol\psi=0$ and $\ol{\partial_0 \psi}=0$,
  \item[(v)] $\ol A_{\mu\nu} =0$,  $\ol{\nabla_0 A_{\mu\nu}} =0$ and $\ol{\nabla_0\nabla_0 A_{\mu\nu}} =0$,
  \item[(vi)] $\ol B_{\mu\nu}=0$ and $\ol{\nabla_0 B_{\mu\nu}}=0$.
 \end{enumerate}
Moreover, $\ol {\hat X}=\ol X$,  $\ol{\nabla_0 \hat X}=\ol{\nabla_0 X}$, $\ol{\nabla_{\mu}\hat X^{\mu}} = \frac{1}{4}\ol Y$ and $\ol{\nabla_0\nabla_{\mu}\hat X^{\mu}} = \frac{1}{4}\ol{\nabla_0\hat Y}$.
\end{theorem}

\subsection{A special case: $\Theta=1$}

Let us briefly discuss the case where the conformal factor $\Theta$ is identical to one,
\begin{equation*}
 \Theta =1
\;,
\end{equation*}
so that the unphysical space-time can be identified with the physical one.
Then the MCFE imply
\begin{equation*}
 s=\frac{1}{6}\lambda \;, \quad L_{\mu\nu}=sg_{\mu\nu}\;, \quad R_{\mu\nu} = \lambda g_{\mu\nu}
\;,
\end{equation*}
i.e.\ the vacuum Einstein equations hold.
We consider the conditions (i)-(vi) of Theorem~\ref{KID_eqns_main} in this setting.
Condition (iii) is equivalent to $\ol Y=0$ and $\ol{\partial_0 Y}=0$, which provide the initial data  for the wave equation (ii). The only solution is $Y=0$,
i.e.\ $X$ needs to be a Killing field, as desired.
Condition (iv) is then automatically satisfied.
Since
\begin{equation}
   B_{\mu\nu} = \mcL_X L_{\mu\nu}  = s\mcL_X g_{\mu\nu} = 2s\nabla_{(\mu}X_{\nu)}
   \;,
\label{Theta_1_B}
\end{equation}
the validity of (vi) follows from (v), and we are left with the conditions
\begin{eqnarray}
 \Box_g X_{\mu} + \lambda X_{\mu} &=& 0
 \;,
\label{Theta_1_1}
\\
 \ol{\nabla_{(\mu}X_{\nu)}}&=& 0
\label{Theta_1_2}
 \;,
\\
 \ol{\nabla_0\nabla_{(\mu}X_{\nu)}}&=& 0
\label{Theta_1_3}
 \;,
\\
 \ol{\nabla_0\nabla_0\nabla_{(\mu}X_{\nu)}}&=& 0
\label{Theta_1_4}
 \;.
\end{eqnarray}
Note that $\ol B_{\mu\nu}=0$ due to  \eq{Theta_1_B} and \eq{Theta_1_2}, so that \eq{Theta_1_1}-\eq{Theta_1_3} imply via the trace of \eq{waveeqn_A} on $\mathcal{H}$ the validity of \eq{Theta_1_4}.

The equations  \eq{Theta_1_1}-\eq{Theta_1_3} form a possible starting point to derive the KID equations on Cauchy surfaces in space-times satisfying the vacuum Einstein equations (cf.\ \cite{beig,moncrief}).

\subsection{A stronger version of Theorem~\ref{KID_eqns_main}}

Let us now investigate to what extent the conditions (iii)-(vi) in Theorem~\ref{KID_eqns_main} imply each other.
For this purpose we choose adapted coordinates $(x^0\equiv t, x^i)$ in the sense that the initial surface is (locally) given by the set $\{x^0=0\}$
and that, on $\mathcal{H}$, the metric takes the form
\begin{equation}
 g|_{\mathcal{H}} = -(\mathrm{d}t)^2 + \ol g_{ij} \mathrm{d}x^i\mathrm{d}x^j = -(\mathrm{d}t)^2 + h_{ij} \mathrm{d}x^i\mathrm{d}x^j
\;,
\label{form_adapted_coord}
\end{equation}
with $h_{ij}$ some Riemannian metric.
Moreover,  we denote by $f$, $f_i$ and $f_{ij}$ generic functions which depend on the indicated fields (and possibly spatial derivatives thereof) and vanish whenever all their arguments vanish.
The symbol $\breve{.}$ is used to denote the $h$-trace-free part of the corresponding 2-rank tensor on $\mathcal{H}$, i.e.\
\begin{equation}
 \breve v_{ij} := v_{ij} - \frac{1}{3} h_{ij} h^{kl}v_{kl}
\;.
\end{equation}

We start with the identity \cite{ttp2}
\begin{eqnarray}
   \nabla_{\nu}A_{\mu}{}^{\nu}
   -\frac{1}{2}\nabla_{\mu}A_{\nu}{}^{\nu}
\,\equiv \,
 \Box_g  X_{\mu}  + R_{\mu}{}^{\nu}X_{\nu}+ 2\nabla_{\mu} Y
   \;.
   \label{waveX_divergenceA}
\end{eqnarray}
Because of \eq{wave_X0} the right-hand side vanishes and we obtain
\begin{eqnarray}
 \ol{  \nabla_{0}A_{00}} &=& 2 \ol g^{kl}\nabla_k \ol A_{0l}
       -\ol g^{kl}\ol{\nabla_{0}A_{kl} }
\,=\,   -\ol g^{kl}\ol{\nabla_{0}A_{kl} } + f(\ol A_{\mu\nu})
   \;,
\label{useful_eqn_1}
\\
   \ol{\nabla_{0}A_{0i}} &=& \frac{1}{2}\nabla_{i}\ol A_{00}  +  \ol g^{kl}\nabla_{k}\ol A_{il}
      -\frac{1}{2}\ol g^{kl}\nabla_{i}\ol A_{kl}
\,=\, f_i(\ol A_{\mu\nu})
   \;,
\label{useful_eqn_2}
\\
 \ol{ \nabla_0 \nabla_{0}A_{00}} &=& 2 \ol g^{kl}\ol{\nabla_0\nabla_k A_{0l}}
       -\ol g^{kl}\ol{\nabla_0\nabla_{0}A_{kl} }
\nonumber
\\
  &=& 2 \ol g^{kl}\nabla_k\ol{\nabla_0 A_{0l}}
       -\ol g^{kl}\ol{\nabla_0\nabla_{0}A_{kl} }
 + f(\ol A_{\mu\nu})
   \;,
\label{useful_eqn_3}
\\
   \ol{\nabla_0\nabla_{0}A_{0i}} &=& \frac{1}{2}\ol{\nabla_0\nabla_{i}A_{00}}  +  \ol g^{kl}\ol{\nabla_0\nabla_{k}A_{il}}
      -\frac{1}{2}\ol g^{kl}\ol{\nabla_0\nabla_{i}A_{kl} }
\nonumber
\\
 &=& \frac{1}{2}\nabla_{i}\ol{\nabla_0A_{00}}  +  \ol g^{kl}\nabla_{k}\ol{\nabla_0A_{il}}
      -\frac{1}{2}\ol g^{kl}\nabla_{i}\ol{\nabla_0A_{kl}} +f_i(\ol A_{\mu\nu})
   \;.
\label{useful_eqn_4}
\end{eqnarray}
We further  have the identity \cite{ttp2}
\begin{eqnarray*}
 \lefteqn{\nabla_{\nu} B_{\mu}{}^{\nu}- \frac{1}{2}\nabla_{\mu}B_{\nu}{}^{\nu} \equiv
  A_{\alpha\beta}(\nabla^{\alpha}L_{\mu}{}^{\beta} - \frac{1}{2}\nabla_{\mu}L^{\alpha\beta})}
\\
&&+L_{\mu}{}^{\kappa}(\Box_g X_{\kappa} + R_{\kappa}{}^{\alpha}X_{\alpha}+ 2\nabla_{\kappa}Y)
  + \frac{1}{2} \nabla_{\mu}(\Box_g Y+ \frac{1}{6}X^{\nu}\nabla_{\nu}R+ \frac{1}{3} R Y)
   \;.
\end{eqnarray*}
With \eq{wave_X0} and \eq{wave_Y0} we deduce
\begin{eqnarray}
 \ol{  \nabla_{0}B_{00}} &=& 2 \ol g^{kl}\nabla_k \ol B_{0l}
       -\ol g^{kl}\ol{\nabla_{0}B_{kl} } + f(\ol A_{\mu\nu})
\nonumber
\\
 &=&   -\ol g^{kl}\ol{\nabla_{0}B_{kl} } + f(\ol A_{\mu\nu},\ol B_{\mu\nu})
   \;,
\label{useful_eqn_6}
\\
   \ol{\nabla_{0}B_{0i}} &=& \frac{1}{2}\nabla_{i}\ol B_{00}  +  \ol g^{kl}\nabla_{k}\ol B_{il}
      -\frac{1}{2}\ol g^{kl}\nabla_{i}\ol B_{kl} +f_i(\ol A_{\mu\nu})
\nonumber
\\
&=& f_i(\ol A_{\mu\nu},\ol B_{\mu\nu})
   \;.
\label{useful_eqn_7}
\end{eqnarray}
Evaluation of \eq{waveeqn_A} on the initial surface gives with  \eq{christoffel1}-\eq{christoffel2}
\begin{eqnarray}
  \ol {\nabla_0\nabla_0 A_{ij}}    &=&  4\ol B_{ij}-  \ol g^{kl}\ol\Gamma^0_{kl}\ol {\nabla_0A_{ij}} + f_{ij}(\ol A_{\mu\nu})
 \;,
\label{useful_eqn_8}
\\
  \ol {\nabla_0\nabla_0 A_{0i}}    &=&  4 \ol B_{0i}-  \ol g^{kl}\ol\Gamma^0_{kl}\ol{\nabla_0 A_{0i}} + f_{i}(\ol A_{\mu\nu})
 \;,
\label{useful_eqn_9}
\\
 \ol { \nabla_0\nabla_0 A_{00} }   &=&  4 \ol B_{00}-  \ol g^{kl}\ol\Gamma^0_{kl}\ol{\nabla_0 A_{00}} + f(\ol A_{\mu\nu})
\;.
\label{useful_eqn_10}
\end{eqnarray}
%
%
%
%
From the definition of $B_{\mu\nu}$ we obtain with \eq{wave_Y0} (set $B:= g^{\mu\nu}B_{\mu\nu}$)
\begin{eqnarray}
 \ol B &\equiv & \ol L^{\mu\nu}\ol A_{\mu\nu}   + \ol{\Box_g Y}
   +   \frac{1}{6}  \ol X^{\mu}\ol{\nabla_{\mu}R} + \frac{1}{3} \ol{RY}
\nonumber
\\
 &=& \ol L^{\mu\nu}\ol A_{\mu\nu}
\;,
\label{useful_eqn_11}
\\
   \ol {\nabla_0 B } & \equiv & \ol{\nabla_0(L^{\mu\nu}A_{\mu\nu}) }   +  \ol{\nabla_0( \Box_g Y  + \frac{1}{6}X^{\mu}\nabla_{\mu}R + \frac{1}{3} RY) }
\nonumber
\\
 &=&   \ol{\nabla_0(L^{\mu\nu}A_{\mu\nu}) }
\;.
\label{useful_eqn_12}
\end{eqnarray}

We use the equations \eq{useful_eqn_1}-\eq{useful_eqn_12} to establish a stronger version of Theorem~\ref{KID_eqns_main}.
Let us assume that
\begin{equation}
 \ol A_{\mu\nu}=0\;, \quad \ol{\nabla_0A_{ij}}=0\;, \quad \breve{\ol B}_{ij}=0\;, \quad (\ol{\nabla_0B_{ij}})\breve{}=0
\;.
\end{equation}
Then by \eq{useful_eqn_1} and \eq{useful_eqn_2} we have $\ol{\nabla_0 A_{\mu\nu}}=0$.
From \eq{useful_eqn_11} and \eq{useful_eqn_12} we  deduce $\ol B = \ol{\nabla_0 B}=0$.
The equations \eq{useful_eqn_3}, \eq{useful_eqn_8} and \eq{useful_eqn_10} yield the system
\begin{eqnarray*}
 \ol{\nabla_0\nabla_0A_{00}} &=& -\ol g^{ij}\ol{\nabla_0\nabla_0 A_{ij}}
\;,
\\
 \ol g^{ij}\ol{\nabla_0\nabla_0 A_{ij}} &=& 4\ol g^{ij} \ol B_{ij} \overset{\ol B=0}{=} 4\ol B_{00}
 \;,
\\
  \ol{\nabla_0\nabla_0A_{00}} &=& 4\ol B_{00}
 \;,
\end{eqnarray*}
from which we conclude  $ \ol{\nabla_0\nabla_0A_{00}} =   \ol g^{ij}\ol{\nabla_0\nabla_0 A_{ij}} = \ol B_{00}=0$.
From \eq{useful_eqn_4} and the trace-free part of \eq{useful_eqn_8} we then deduce $\ol{\nabla_0\nabla_0 A_{\mu\nu}}=0$,
and the equations \eq{useful_eqn_9} and \eq{useful_eqn_11} imply $\ol B_{\mu\nu}=0$.
Moreover,  invoking \eq{useful_eqn_6} and \eq{useful_eqn_12} yields
\begin{eqnarray*}
& \ol{\nabla_0B_{00}} \,=\, -\ol g^{ij}\ol{\nabla_0 B_{ij}}
\;,&
\\
& 0\,=\, \ol{\nabla_0 B}\,=\, \ol g^{ij}\ol{\nabla_0 B_{ij}} - \ol{\nabla_0 B_{00}}
\;,&
\end{eqnarray*}
i.e.\   $\ol{\nabla_0B_{00}} =\ol g^{ij}\ol{\nabla_0 B_{ij}}=0$.
The equation \eq{useful_eqn_7} then completes the proof that $\ol{\nabla_0 B_{\mu\nu}}=0$.

%
%

We end up with the result
\begin{theorem}
 \label{KID_eqns_main2}
Assume we have been given, in $3+1$ dimensions,  an``unphysical'' space-time $(\mcM, g, \Theta)$, with ($g, \Theta$) a smooth solution of the MCFE \eq{conf1}-\eq{conf6}.
Consider a spacelike hypersurface $\mathcal{H}\subset \mcM$.
 Then there exists a vector field $\hat X$ satisfying the unphysical Killing equations \eq{2_conditions} on $\mathrm{D}^+(\mathcal{H})$
if and only if there exists a pair $(X,Y)$, $X$ a vector field and $Y$ a function, which fulfills the KID equations, i.e.
 \begin{enumerate}
  \item[(a)] equations (i)-(iv) of Theorem~\ref{KID_eqns_main},
  \item[(b)] $\ol A_{\mu\nu} =0$ and  $\ol{\nabla_0 A_{ij}} =0$   with
 $A_{\mu\nu} \equiv 2\nabla_{(\mu}X_{\nu)}  - 2Y g_{\mu\nu}$,
  \item[(c)] $\breve{\ol B}_{ij}=0$ and $(\ol{\nabla_0 B_{ij}})\breve{}=0$ with
$B_{\mu\nu}\equiv \mcL_XL_{\mu\nu} +   \nabla_{\mu}\nabla_{\nu}Y$.
 \end{enumerate}
Moreover, $\ol {\hat X}=\ol X$,  $\ol{\nabla_0 \hat X}=\ol{\nabla_0 X}$, $\ol{\nabla_{\mu}\hat X^{\mu}} = \frac{1}{4}\ol Y$ and $\ol{\nabla_0\nabla_{\mu}\hat X^{\mu}} = \frac{1}{4}\ol{\nabla_0\hat Y}$.
\end{theorem}

\subsection{The (proper) KID equations}

We want to replace the equations $\ol{\partial_0\psi}=0$ and $(\ol{\nabla_0 B_{ij}})\breve{}=0$ appearing in
Theorem~\ref{KID_eqns_main2}
 by intrinsic equations on $\mathcal{H}$ in the sense that they involve at most first-order transverse derivatives of $X$ and $Y$, which belong to the freely prescribable initial data for the wave equations \eq{wave_X0} and \eq{wave_Y0}.
The higher-order derivatives appearing can be eliminated via
\eq{wave_Y0} which implies
\begin{eqnarray}
  \ol{\nabla_0\nabla_0 Y }&=&  \ol g^{kl}\ol{\nabla_k\nabla_l Y}+ \frac{1}{6}\ol X^{\mu}\ol{\nabla_{\mu}R}+ \frac{1}{3} \ol R \ol Y
\;.
\label{tt-Y}
\end{eqnarray}
%
We are straightforwardly led to

\begin{theorem}
\label{KID_eqns_main_scri2}
Assume that we have been given a $3+1$-dimensional space-time $(\mcM, g, \Theta)$, with ($g, \Theta$) being a smooth solution of the MCFE.
Let $\mathring X$ and $\mathring\Lambda$ be spacetime vector fields,
and $\mathring Y$ and $\mathring\Upsilon$
be functions defined along a spacelike hypersurface $\mathcal{H}\subset \mcM$. 
Then there exists a smooth space-time vector field $X$ with $\ol X=\mathring X$,  $\ol{\nabla_0 X}=\mathring \Lambda$, $\ol{\nabla_{\mu}X^{\mu}} = \frac{1}{4}\mathring Y$ and $\ol{\nabla_0\nabla_{\mu}X^{\mu}} = \frac{1}{4}\mathring \Upsilon$  satisfying the unphysical Killing equations   \eq{2_conditions} on $\mathrm{D}^+(\mathcal{H})$
(and thus corresponding to a Killing field of the physical space-time)
if and only if in the adapted coordinates \eq{form_adapted_coord}:
 \begin{enumerate}
   \item[(i)] $ \ol\phi\equiv   \mathring X^{\mu}\ol{\nabla_{\mu}\Theta} - \ol \Theta \mathring Y=0$,
\\
 $\ol{\partial_0\phi} \equiv   \mathring \Lambda^{\mu}\ol{\nabla_{\mu}\Theta}+  \mathring X^{\mu}\ol{\nabla_{\mu}\nabla_0\Theta}
  - \ol \Theta \mathring \Upsilon   - \ol {\nabla_0\Theta} \mathring Y =0$,
  \item[(ii)] $\ol \psi \equiv \mathring X^{\mu}\ol{\nabla_{\mu}s}    +\ol s\mathring Y - \ol{\nabla^{i}\Theta} \tilde \nabla_{i}\mathring Y
  + \ol{\nabla_{0}\Theta} \mathring \Upsilon =0 $,
\\
   $\ol{\partial_0 \psi}^{\mathrm{intr}} := \mathring \Lambda^{\mu}\ol{\nabla_{\mu}s}  + \mathring  X^{\mu}\ol{\nabla_{\mu}\nabla_0s }
    +\ol{\nabla_0s}\mathring Y +(\ol s +\ol{\nabla_{0}\nabla_0\Theta}  )\mathring \Upsilon
  -\ol{\nabla^{i}\nabla_0\Theta}\tilde \nabla_{i}\mathring Y
 +\ol{\nabla_{0}\Theta }(\Delta_h \mathring Y - \ol \Gamma^{k}_{0k}\mathring \Upsilon
+ \frac{1}{6}\mathring X^{\mu}\ol{\nabla_{\mu}R}+ \frac{1}{3} \ol R \mathring Y )
 - \ol{\nabla^{k}\Theta }(\tilde \nabla_{k}\mathring \Upsilon - \ol \Gamma^{i}_{0k}\tilde\nabla_{i}\mathring Y ) =0 $,
  \item[(iii)] $\ol A_{ij} \equiv  2\nabla_{(i}\mathring X_{j)} - 2\mathring Y \ol g_{ij}=0$,
\\
$\ol A_{0i} \equiv  \mathring \Lambda_{i} + \nabla_{i} \mathring X_{0} =0$,
\\
$\ol A_{00} \equiv 2 \mathring \Lambda_{0} + 2\mathring Y =0$,
\\
$\ol {\nabla_0 A_{ij}} \equiv 2\tilde\nabla_{(i}\mathring \Lambda_{j)}
-2\ol\Gamma^{k}_{0(i}\nabla_{k}\mathring X_{j)}
-2\ol\Gamma^{0}_{ij}\mathring \Lambda_{0}
+ 2\ol R_{0(ij)}{}^{\mu}\mathring X_{\mu}- 2\mathring \Upsilon \ol g_{ij}=0$,
  \item[(iv)]  $\breve{\ol B}_{ij}\equiv (\mathring X^{\mu}\ol{\nabla_{\mu}L_{ij}} +2 \ol L_{\mu(i}\nabla_{j)}\mathring X^{\mu} +  \tilde\nabla_i\tilde\nabla_j\mathring Y-\ol \Gamma^{0}_{ij}\mathring \Upsilon  )\breve{}=0$,
\\
 $(\ol{\nabla_0 B^{\mathrm{intr}}_{ij}})\breve{} := [  \ol{\mcL_{\mathring X}\nabla_0L_{ij}}
+ 2\ol L_{\mu(i}(\partial_{j)}\mathring \Lambda^{\mu} + \ol \Gamma^{\mu}_{j\alpha}\mathring\Lambda^{\alpha}
-\ol \Gamma^k_{0j}\nabla_k\mathring X^{\mu})
+ 2\ol L_{k(i}\ol R_{j)\mu 0}{}^k \mathring X^{\mu}
 +  \tilde \nabla_{i}\tilde\nabla_{j} \mathring \Upsilon
- \ol\Gamma^0_{ij}(\Delta_h \mathring Y + \frac{1}{6}\mathring X^{\mu}\ol{\nabla_{\mu}R}+ \frac{1}{3} \ol R \mathring Y )
- 2\ol\Gamma^{k}_{0(i}\tilde\nabla_{j)}\tilde \nabla_{k}\mathring Y
 + ( \ol R_{0ij}{}^{0} + \ol\Gamma^{k}_{0i}\ol\Gamma^0_{jk} + \ol\Gamma^0_{ij}\ol \Gamma^k_{0k})\mathring \Upsilon
+ (\ol R_{0ij}{}^{k}\ -\tilde\nabla_i \ol \Gamma^k_{0j})\tilde\nabla_k\mathring Y  ]\breve{}=0$.
 \end{enumerate}
\end{theorem}

\begin{proof}
Assume that there exist fields $\mathring X$, $\mathring \Lambda$, $\mathring Y$ and $\mathring \Upsilon$ which satisfy (i)-(iv).
These fields provide the initial data for the wave equations \eq{wave_X0} and \eq{wave_Y0} for $X$ and $Y$.
A solution exists due to standard results. Once \eq{wave_X0} and \eq{wave_Y0} are satisfied the considerations above reveal that
(i)-(iv) are equivalent to (a)-(c) of Theorem~\ref{KID_eqns_main2}, i.e.\ all the hypotheses of Theorem~\ref{KID_eqns_main2} hold
and we are done.
From the derivation of (i)-(iv) it follows that these conditions are necessary, as well.
\qed
\end{proof}

\begin{remark}
{\rm
 We call the equations in (i)-(iv) the \textit{(proper) KID equations on~$\mathcal{H}$}.
}
\end{remark}

\section{KID equations on a spacelike $\scri^-$}
\label{sec_scri}

\subsection{Derivation of the (reduced) KID equations}

Let us restrict now attention to space-times which contain a spacelike $\scri^-$, which we take henceforth  as initial surface (recall that this requires a positive
cosmological constant  $\lambda$).
We impose the $(R=0, \ol s=0,  \ol g_{00}=-1, \ol g_{0i}=0, \hat g_{\mu\nu} = \ol g_{\mu\nu})$-wave-map gauge condition introduced in Section~\ref{subsec_gauge}.
Recall that the freely prescribable  data on $\scri^-$ for the Cauchy problem are the conformal class of a Riemannian metric $h_{ij}$ and a
symmetric, trace- and divergence-free tensor $D_{ij}$. The MCFE then imply the  constraint equations \eq{constr2}-\eq{constr8} on $\scri^-$.
%
In Appendix~\ref{app_alternative} it is shown that a solution to the MCFE 
further satisfies 
\begin{eqnarray}
 \ol{\nabla_0\nabla_0\Theta}=0\;, \quad \ol R_{0ij}{}^k=0
\;.
\end{eqnarray}
We are now ready to evaluate the conditions (i)-(iv) of Theorem~\ref{KID_eqns_main_scri2}.

The condition (i) becomes
\begin{eqnarray}
 \mathring X^0 =0\;, \quad \mathring \Lambda^0 = \mathring Y
\;.
\end{eqnarray}
Then condition (ii) is satisfied iff (set $\Delta_{\tilde g}:=\ol g^{ij}\tilde\nabla_i\tilde\nabla_j$)
\begin{eqnarray}
 \mathring \Upsilon =0\;, \quad \mathring X^i \tilde\nabla_i  \tilde R + 2 \tilde R\mathring Y
 + 4\Delta_{\tilde g} \mathring Y=0
\;.
\label{second_cond}
\end{eqnarray}
The condition $\ol A_{\mu\nu}=0$ requires
\begin{eqnarray}
 \mathring \Lambda^i  &=& 0\;,
\\
 \mathring Y &=& \frac{1}{3}\tilde\nabla_i\mathring X^i \;,
\label{Y_div}
\\
 ( \tilde\nabla_{(i}\mathring X_{j)})\breve{} &=& 0
\label{conf_Killing}
\;.
\end{eqnarray}
The condition $\ol{\nabla_0A_{ij}}=0$ is then automatically fulfilled.

We reconsider the second condition in \eq{second_cond}.
Observe that \eq{Y_div}, \eq{conf_Killing} and the second Bianchi identity imply the relation
\begin{eqnarray*}
0 \,=\, \tilde\nabla^i\tilde\nabla^j\ol A_{ij} &=&  \tilde\nabla_i\Delta_{\tilde g}\mathring X^i + \Delta_{\tilde g} \mathring Y   + \frac{1}{2}\mathring X^i \tilde\nabla_i\tilde R
 + \underbrace{\tilde R_{jk}\tilde\nabla^j\mathring X^k}_{=\tilde R\mathring Y}
\\
&=& 4 \Delta_{\tilde g} \mathring Y   + \mathring X^i \tilde\nabla_i\tilde R
 + 2\tilde R\mathring Y
\;,
\end{eqnarray*}
i.e.\ \eq{second_cond} follows from  \eq{Y_div} and \eq{conf_Killing}.

We have
\begin{eqnarray*}
   \breve{\ol B}_{ij} &=& (\mathring X^{k}\tilde \nabla_{k}\tilde L_{ij}+2 \tilde L_{k(i}\tilde \nabla_{j)}\mathring X^{k} +   \tilde\nabla_{i}\tilde\nabla_{j}\mathring Y)\breve{}
\\
 &=& \mcL_{\mathring X^k\partial_k}\breve{\tilde L}_{ij}+   (\tilde\nabla_{i}\tilde\nabla_{j}\mathring Y)\breve{}
\:,
\end{eqnarray*}
and
\begin{eqnarray*}
 (\ol{\nabla_0 B^{\mathrm{intr}}_{ij}})\breve{} &=&  -\sqrt{\frac{\lambda}{3}}(D_{ij} \mathring Y
 +  \mathring X^{k}\tilde\nabla_{k}D_{ij}
 +2D_{k(i}\tilde\nabla_{j)}\mathring X^{k})
\\
  &=&  -\sqrt{\frac{\lambda}{3}}(\mcL_{\mathring X^k\partial_k}D_{ij} +D_{ij} \mathring Y)
\;.
\end{eqnarray*}
We observe that due to the second Bianchi identity and \eq{Y_div}
\begin{eqnarray*}
 \tilde\nabla_i\tilde\nabla^k \ol A_{jk} &=&  \mcL_{\mathring X^k\partial_k}\tilde R_{\mu\nu} + \tilde\nabla_i \tilde\nabla_j\mathring Y  + 2\mathring X^k\tilde\nabla_{[i}\tilde R_{j]k}
+ \Delta_{\tilde g}\tilde\nabla_i\mathring X_j
\\
&& + 2\tilde R_{i}{}^{k}{}_{j}{}^l  \tilde\nabla_k\mathring X_l  - 2 \tilde R_{ij}\mathring Y - \tilde R_{i}{}^{k}\ol A_{jk}
\;.
\end{eqnarray*}
Symmetrizing this expression, taking its traceless part and taking $\ol A_{ij}=0$ into account we end up with
\begin{equation*}
    \mcL_{\mathring X^k\partial_k}\breve{\tilde L}_{\mu\nu} + (\tilde\nabla_i \tilde\nabla_j\mathring Y  )\breve{}
  =0
\;,
\end{equation*}
i.e.\ $\breve{\ol B}_{ij}$ holds automatically, as well.

\begin{theorem}
\label{KID_eqns_main_hyp_infinity}
Assume we have been given a $3+1$-dimensional ``unphysical'' space-time $(\mcM, g, \Theta)$, with  $(g_{\mu\nu}, \Theta, s, L_{\mu\nu}, d_{\mu\nu\sigma}{}^{\rho})$ a smooth solution of the MCFE with $\lambda>0$  in the $(R=0, \ol s=0,  \ol g_{00}=-1, \ol g_{0i}=0, \hat g_{\mu\nu} = \ol g_{\mu\nu})$-wave-map gauge.
Then there exists a smooth vector field $X$
satisfying the unphysical Killing equations \eq{2_conditions} on $\mathrm{D}^+(\scri^-)$
(and thus corresponding to a Killing vector field of the physical space-time)
if and only if there exists a conformal Killing vector field $\mathring X$ on $(\scri^-, \tilde g=h_{ij}\mathrm{d}x^i\mathrm{d}x^j)$ such that the reduced KID equations
\begin{eqnarray}
   \mcL_{\mathring X}D_{ij} + \frac{1}{3}D_{ij}\tilde\nabla_k\mathring X^k &=&0
\label{reduced_KID_2}
\end{eqnarray}
hold (recall that the symmetric, trace- and divergence-free tensor  field $D_{ij}= \ol d_{0i0j}$ belongs to the freely prescribable initial data).
In that case $X$ satisfies
\begin{equation}
 \ol X^0=0\;, \quad \ol X^i = \mathring X^i\;, \quad  \ol {\nabla_0 X^0} = \frac{1}{3}\tilde\nabla_i\mathring X^i\;, \quad \ol{\nabla_0 X^i}=0\;.
\end{equation}
\end{theorem}

\begin{remark}
\em{
Note that, in contrast to the $\lambda=0$-case treated in \cite{ttp2}, the candidate fields,
i.e.\ the conformal Killing fields on $\scri^-$, do depend here on the initial data $h=h_{ij}\mathrm{d}x^i\mathrm{d}x^j$.
}
\end{remark}

\begin{remark}
\em{ 
For initial data with $D_{ij}=0$ the reduced Killing equations \eq{reduced_KID_2} are always satisfied, and each candidate field, i.e.\ each conformal
Killing field on the initial manifold, extends to a Killing field of the physical space-time.
}
\end{remark}

In terms of an initial value problem Theorem~\ref{thm_well-posedness} and \ref{KID_eqns_main_hyp_infinity} state that given a Riemannian manifold $(\mathcal{H},h)$ and
a symmetric, trace- and divergence-free tensor field $D_{ij}$ there exists an (up to isometries) unique evolution into a space-time manifold $(\mcM,g,\Theta)$ with $\mathcal{H}=\scri^-$, $\ol g_{ij}=h_{ij}$ and $\ol d_{0i0j}=D_{ij}$ which fulfills the MCFE and contains a vector field satisfying the unphysical Killing equations \eq{2_conditions} if and only if there exists a
conformal Killing vector field $\mathring X$ on $(\mathcal{H},h)$ such that the reduced KID equations \eq{reduced_KID_2} hold.

\subsection{Properties of the reduced KID equations}

We compute how the reduced KID equations \eq{reduced_KID_2} behave under conformal transformations.
For this consider the conformally rescaled metric $\tilde{\tilde g}:=\Omega^2 \tilde g$ with $\Omega$ some positive function.
Expressed in terms of $\tilde{\tilde g}$
\eq{reduced_KID_2} becomes
%
\begin{eqnarray}
  \mcL_{\mathring X}(\Omega^{-1} D_{ij}) +  \frac{1}{3}(\Omega^{-1}  D_{ij})\tilde {\tilde\nabla}_k\mathring X^k &=&0
\;,
\label{reduced_KID_3}
\end{eqnarray}
i.e.\ they are conformally covariant in the following sense:
\begin{lemma}
\label{conf_KIDs}
The pair $(\tilde g_{ij}, D_{ij})$ is a solution of  the reduced KID equations \eq{reduced_KID_2} if and only if
the conformally rescaled pair  $(\Omega^2\tilde g_{ij}, \Omega^{-1}D_{ij})$, with $\Omega$ some positive function,
is a solution of these equations.
\end{lemma}

This is consistent with the observation that conformal rescalings of the initial data do not change the isometry class
of the emerging space-time.

\subsection{Some special cases}

Let us finish by taking a look at some special cases:

\subsubsection{Compact initial manifolds}

We  consider  a \textit{compact} initial manifold $(\scri^-, \tilde g)$
and assume that it admits a conformal Killing field $\mathring X$.
Then there exists (cf.\ e.g.\ \cite{kuehnel}) a positive function~$\Omega$ such that the  conformally rescaled metric $\tilde{\tilde g}=\Omega^2 \tilde g$  has one of the following properties:
\begin{itemize}
\item Either  $(\scri^-, \tilde{\tilde g})=(S^3,s_{ij}\mathrm{d}x^i\mathrm{d}x^j )$ is the standard 3-sphere,
\item or $\mathring X$ is a Killing vector field w.r.t.\ $\tilde{\tilde g}$.
\end{itemize}
If $(\scri^-, \tilde{\tilde g})$ is the round 3-sphere all the conformal Killing fields are explicitly known.
In the second case where $\mathring X$ is a Killing vector field w.r.t.\ $\tilde{\tilde g}$ the equation \eq{reduced_KID_3}
simplifies to
\begin{eqnarray}
   \mcL_{\mathring X}(\Omega^{-1} D_{ij}) &=&0
\;.
\label{reduced_KID_4}
\end{eqnarray}
That implies:
\begin{lemma}
 Consider a solution  of the vacuum Einstein equations which admits a  compact spacelike $\scri^-$ and has a non-trivial Killing field. If  $(\scri^-, \tilde g)$ is not conformal to a standard 3-sphere, then there exists a choice of conformal factor so that space-time   Killing vector corresponds to a Killing field (rather than a conformal Killing field) of $(\scri^-, \tilde g)$.
\end{lemma}

\subsubsection{Maximally symmetric space-times}
\label{max_symm}

Let us consider the case where the initial manifold  admits the maximal number of conformal Killing vector fields.
Clearly this is a prerequisite to obtain a maximally symmetric space-time once the evolution problem has been solved.
A connected  3-dimensional Riemannian manifold $(\mathcal{H},h)$ admits at most 10 linearly independent conformal Killing vector fields. 
If equality is attained, $(\mathcal{H},h)$ is known to be locally conformally flat  \cite{SemmelmannHab}.

Let us first consider the compact case. We use a classical result due to Kuiper (cf.\ \cite{kuehnel}):
\begin{theorem} 
 For any $n$-dimensional, simply connected, conformally flat Riemann manifold $(\mathcal{H},h)$, there exists
a conformal immersion $(\mathcal{H},h) \hookrightarrow (S^n, s=s_{ij}\mathrm{d}x^i\mathrm{d}x^j)$, the so-called developing map, which is unique up
to composition with M\"obius transformations.
If $\mathcal{H}$ is compact this map defines a conformal diffeomorphism from $(\mathcal{H},h) $ onto $ (S^n, s)$.
\end{theorem}
Since only the conformal class of the initial manifold  matters we thus may assume  $(\mathcal{H},h)$ for compact $\mathcal{H}$   to be the standard 3-sphere from the outset.
To end up with a maximally symmetric physical space-time containing 10 
independent Killing  fields
one needs to make sure that each of the conformal Killing fields extends to a space-time vector field satisfying the unphysical Killing equations \eq{2_conditions}.
In other words one needs to choose $D_{ij}$ such that the reduced KID equations  \eq{reduced_KID_2}  hold for each and every conformal Killing field on $(S^3,s)$. 
Via a stereographic projection onto Euclidean space one shows that this is only possible
when $D_{ij}$ is proportional to the round sphere metric. But $D_{ij}$ is traceless, and thus needs to vanish.
%
For  data  $(\mathcal{H},h)=(S^3, s)$ and $D_{ij}=0$  one ends up with de Sitter space-time.
This is in accordance with the fact that de Sitter space-time is (up to isometries) the unique maximally symmetric, complete space-time with positive scalar curvature.

The non-compact case is somewhat more involved since the developing map does in general not define a global conformal diffeomorphism
into $(S^n,s)$. For convenience let us therefore make some simplifying assumptions on $(\mathcal{H},h)$ which allow us to apply a result
by Schoen \& Yau \cite{schoen} (we restrict attention to 3 dimensions when stating it):
\begin{theorem}
 Let $(\mathcal{H},h)$ be a complete, simply connected, conformally flat 3-dimensional Riemannian manifold and $\Phi: \mathcal{H} \hookrightarrow S^3$ its developing map.
Assume that $|R(h)|$ is bounded  on $\mathcal{H}$ and that  $d(\mathcal{H})<\frac{1}{3}$.%
\footnote{
For the definition of the   invariant $d(\mathcal{H})$  in terms of the minimal Green's function for the conformal Laplacian  we refer the reader to \cite{schoen}.
}
Then $\Phi$ is one-to-one and gives a conformal diffeomorphism from $\mathcal{H}$ onto  a simply connected domain of~$S^3$.
\end{theorem}
%
We conclude, again, that the emerging space-time will be maximally symmetric iff $D_{ij}=0$, and will be (isometric to) a part of de Sitter space-time.

\subsubsection{Non-existence of stationary space-times}

For $(\mcM,g,\Theta)$ to contain a timelike isometry there must exist a vector field $X$ satisfying the unphysical Killing equations \eq{2_conditions}
which is null on $\scri^-$ (it cannot be timelike since $\ol X^0=0$),
\begin{eqnarray*}
0 = \ol g_{\mu\nu} \ol X^{\mu}\ol X^{\nu} = h_{ij} \ol X^{i}\ol X^{j} \quad \Longrightarrow \quad \ol X^i =0
\;.
\end{eqnarray*}
But then the preceding considerations show that $\ol X^{\mu}=\ol{\nabla_0 X^{\mu}}=\ol Y = \ol{\nabla_0 Y}=0$, and solving the wave equations for
$X$ and $Y$, \eq{wave_X0} and \eq{wave_Y0}, yields that $X$ vanishes identically. It follows that
there is no vacuum space-time  with $\lambda>0$ which is stationary near $\scri^-$. (Compare~\cite[Section~4]{F_lambda}.)

\vspace{1.2em}
\noindent {\textbf {Acknowledgements}}
I would like to thank my advisor Piotr T. Chru\'sciel for  useful discussions, comments and for reading a first draft of this article.
Moreover, I am grateful to Helmut Friedrich for pointing out reference \cite{schoen} to me.
Supported in part by the  Austrian Science Fund (FWF): P 24170-N16.

\appendix

\section{Equivalence between the CWE and the MCFE}
\label{app_alternative}

\subsection{Conformal wave equations (CWE)}

In \cite{ttp} the MCFE \eq{conf1}-\eq{conf6} have been rewritten as a system of conformal wave equations (CWE),
\begin{eqnarray}
 \Box^{(H)}_{ g} L_{\mu\nu}&=&  4 L_{\mu\kappa} L_{\nu}{}^{\kappa} -  g_{\mu\nu}| L|^2
  - 2\Theta d_{\mu\sigma\nu}{}^{\rho}  L_{\rho}{}^{\sigma}
 + \frac{1}{6}\nabla_{\mu}\nabla_{\nu} R
  \label{cwe1}
  \;,
  \\
  \Box_g  s  &=& \Theta| L|^2 -\frac{1}{6}\nabla_{\kappa} R\,\nabla^{\kappa}\Theta  - \frac{1}{6} s  R
  \label{cwe2}
  \;,
  \\
  \Box_{ g}\Theta &=& 4 s-\frac{1}{6} \Theta  R
  \label{cwe3}
  \;,
  \\
  \Box^{(H)}_g d_{\mu\nu\sigma\rho}
  &=& \Theta d_{\mu\nu\kappa}{}^{\alpha}d_{\sigma\rho\alpha}{}^{\kappa}
   - 4\Theta d_{\sigma\kappa[\mu} {}^{\alpha}d_{\nu]\alpha\rho}{}^{\kappa}
  + \frac{1}{2}R d_{\mu\nu\sigma\rho}
  \label{cwe4}
  \;,
  \\
  R^{(H)}_{\mu\nu}[g] &=& 2L_{\mu\nu} + \frac{1}{6} R g_{\mu\nu}
  \label{cwe5}
  \;.
\end{eqnarray}
Here
\begin{equation}
 R^{(H)}_{\mu\nu} :=  R_{\mu\nu} -  g_{\sigma(\mu}\hat\nabla_{\nu)} H^{ \sigma}
 \label{ricci_riccired}
 \;,
\end{equation}
denotes the reduced Ricci tensor.
The reduced wave-operator $\Box^{(H)}_g$
(which is needed to obtain a PDE-system  with a diagonal principal part) 
is defined via its action on covector fields $v_{\lambda}$,
\begin{eqnarray}
 \Box_g^{(H)}v_{\lambda} &:=& \Box_g v_{\lambda} - g_{\sigma[\lambda}(\hat\nabla_{\mu]} H^{ \sigma})v^{\mu}
  +(2L_{\mu\lambda} -R^{(H)}_{\mu\lambda} + \frac{1}{6}Rg_{\mu\lambda})v^{\mu}
   \;,
\end{eqnarray}
 and similar formulae hold for higher-valence covariant tensor fields.

In the following we want to show that a solution of the CWE
in the gauge \eq{gauge_conditions_compact} is a solution of the MCFE if and only if the constraint equations \eq{constr1}-\eq{constr8} hold on $\scri^-$,
\begin{eqnarray}
& h^{ij} D_{ij}=0\;, \quad \tilde\nabla^j D_{ij}=0
\;,
\label{constr1*}
&
\\
 &\ol g_{00} =-1\;, \quad \ol g_{0i}=0\;, \quad \ol g_{ij} = h_{ij}\;, \quad \ol{\partial_0 g_{\mu\nu}}=0\;,
\label{constr2*}
&
\\
& \ol \Theta =0\;, \quad \ol{\partial_0 \Theta} = \sqrt{\frac{\lambda}{3}}\;,
\label{constr3*}
&
\\
&\ol s=0\;, \quad \ol{\partial_0 s} =  \sqrt{\frac{\lambda}{48}} \,\tilde R\;,
\label{constr4*}
&
\\
&\ol L_{ij} = \tilde L_{ij} \;, \quad \ol L_{0i}=0\;, \quad \ol L_{00} = \frac{1}{4}\tilde R \;,
\label{constr5*}
&
\\
& \ol{\partial_0 L_{ij}} =  -\sqrt{\frac{\lambda}{3}} \, D_{ij}\;, \quad
  \ol{\partial_0 L_{0i}} =\frac{1}{4}\tilde\nabla_i\tilde R  \;, \quad  \ol{\partial_0 L_{00}} =0\;,
\label{constr6*}
&
\\
&  \ol d_{0i0j} = D_{ij}\;,\quad \ol d_{0ijk}  =  \sqrt{\frac{3}{\lambda}} \tilde C_{ijk}\;,
\label{constr7*}
&
\\
& \ol{\partial_{0} d_{0i0j}}= \sqrt{\frac{3}{\lambda}}\tilde B_{ij}\;, \quad
\ol{\partial_{0} d_{0ijk}} = 2\tilde \nabla_{[j}D_{k]i}\;.
\label{constr8*}
 &
\end{eqnarray}

\subsection{An intermediate result}
In close analogy to \cite[Theorem~3.7]{ttp} one establishes the following result:
\begin{theorem}
\label{inter-thm}
Assume we have been given data ($\mathring g_{\mu\nu}$, $\mathring K_{\mu\nu}$, $\mathring s$, $\mathring S$, $\mathring \Theta$, $\mathring \Omega$, $\mathring L_{\mu\nu}$, $\mathring M_{\mu\nu}$, $\mathring d_{\mu\nu\sigma}{}^{\rho}$, $\mathring D_{\mu\nu\sigma}{}^{\rho}$)
on a spacelike hypersurface $\mathcal{H}$  and a gauge source function $R$, such that $\mathring g_{\mu\nu}$ is the restriction to $\mathcal{H}$ of a Lorentzian metric, $\mathring K_{\mu\nu}$,  $\mathring L_{\mu\nu}$ and  $\mathring M_{\mu\nu}$ are symmetric, $ \mathring L \equiv \mathring L_{\mu}{}^{\mu}  = \overline R/6$,  $  \mathring M_{\mu}{}^{\mu} =\ol{ \partial_0 R}/6$, and such that $\mathring d_{\mu\nu\sigma}{}^{\rho}$ and $\mathring D_{\mu\nu\sigma}{}^{\rho}$  satisfy all the algebraic properties of the Weyl tensor.
Suppose further that there exists a solution ($g_{\mu\nu}$, $s$, $\Theta$, $L_{\mu\nu}$, $d_{\mu\nu\sigma}{}^{\rho}$)  of the CWE \eq{cwe1}-\eq{cwe5} with gauge source function $R$ which induces the above data on $\mathcal{H}$,
\begin{eqnarray*}
 &\ol g_{\mu\nu} = \mathring g_{\mu\nu}\;, \quad \ol s = \mathring s \;, \quad \ol\Theta =\mathring \Theta\;, \quad\ol L_{\mu\nu} = \mathring L_{\mu\nu}\;,\quad
\ol d_{\mu\nu\sigma}{}^{\rho} = \mathring  d_{\mu\nu\sigma}{}^{\rho} \;,&
\\
& \ol {\partial_0 g_{\mu\nu}} = \mathring K_{\mu\nu}\;, \quad \ol {\partial_0 s} = \mathring S \;, \quad \ol{\partial_0 \Theta} =\mathring \Omega\;, \quad\ol{\partial_0  L_{\mu\nu}} = \mathring M_{\mu\nu}\;,\quad
\ol {\partial_0 d_{\mu\nu\sigma}{}^{\rho}} = \mathring  D_{\mu\nu\sigma}{}^{\rho} \;,&
\end{eqnarray*}
 and fulfills the following conditions:
\begin{enumerate}
 \item The MCFE \eq{conf1}-\eq{conf4} and their covariant derivatives are fulfilled on $\mathcal{H}$;
 \item  equation \eq{conf5} holds at one point on $\mathcal{H}$;
\item  $\ol W_{\mu\nu\sigma}{}^{\rho}[g]=\ol \Theta \,\ol d_{\mu\nu\sigma}{}^{\rho}$ and
 $\ol {\nabla_{0} W_{\mu\nu\sigma}{}^{\rho}}[g]=\ol{\nabla_{0} (\Theta \, d_{\mu\nu\sigma}{}^{\rho})}$;
 \item  the wave-gauge vector $H^{\sigma}$ and its first- and second-order covariant derivatives $\nabla_{\mu} H^{\sigma}$ and $\nabla_{\mu}\nabla_{\nu}H^{\sigma}$ vanish on $\mathcal{H}$;
\item the covector  field $\zeta_{\mu}\equiv -4(\nabla_{\nu}L_{\mu}{}^{\nu} - \frac{1}{6}\nabla_{\mu}R)$ and its covariant derivative
 $\nabla_{\nu}\zeta_{\mu}$  vanish on $\mathcal{H}$.
\end{enumerate}
Then
\begin{enumerate}
 \item[a)] $H^{\sigma}=0$ and $R_g=R$ (where $R_g$ denotes the Ricci scalar of $g_{\mu\nu}$);
\item[b)] $L_{\mu\nu}$ is the Schouten  tensor of $g_{\mu\nu}$;
\item[c)] $\Theta d_{\mu\nu\sigma}{}^{\rho}$ is the Weyl tensor of $g_{\mu\nu}$;
\item[d)] ($g_{\mu\nu}$, $s$, $\Theta$, $L_{\mu\nu}$, $d_{\mu\nu\sigma}{}^{\rho}$)  solves the MCFE \eq{conf1}-\eq{conf6} in the ($H^{\sigma}=0$, $R_g=R$)-gauge.
\end{enumerate}
The conditions 1.-5.\ are also necessary for d) to be true.
\end{theorem}

\subsection{Applicability of Theorem~\ref{inter-thm} on $\scri^-$}

We now consider the case where $\mathcal{H}=\scri^-$. Using the gauge \eq{gauge_conditions_compact}
we want to show that the hypotheses of Theorem~\ref{inter-thm} are fulfilled
 by any tuple ($g_{\mu\nu}$, $s$, $\Theta$, $L_{\mu\nu}$, $d_{\mu\nu\sigma}{}^{\rho}$)
which satisfies the  constraint equations \eq{constr1*}-\eq{constr8*} and the CWE.

For $R=0$ the CWE reduce to
\begin{eqnarray}
 \Box^{(H)}_{ g} L_{\mu\nu}&=&  4 L_{\mu\kappa} L_{\nu}{}^{\kappa} -  g_{\mu\nu}| L|^2
  - 2\Theta d_{\mu\sigma\nu}{}^{\rho}  L_{\rho}{}^{\sigma}
  \label{cwe1*}
  \;,
  \\
  \Box_g  s  &=& \Theta| L|^2
  \label{cwe2*}
  \;,
  \\
  \Box_{ g}\Theta &=& 4 s
  \label{cwe3*}
  \;,
  \\
  \Box^{(H)}_g d_{\mu\nu\sigma\rho}
  &=& \Theta d_{\mu\nu\kappa}{}^{\alpha}d_{\sigma\rho\alpha}{}^{\kappa}
   - 4\Theta d_{\sigma\kappa[\mu} {}^{\alpha}d_{\nu]\alpha\rho}{}^{\kappa}
  \label{cwe4*}
  \;,
  \\
  R^{(H)}_{\mu\nu}[g] &=& 2L_{\mu\nu}
  \label{cwe5*}
  \;.
\end{eqnarray}

First of all note that 
$ \ol L=0=\ol R/6$ and $\ol{\partial_0 L}=0= \ol{\partial_0 R}/6$, as required.
Moreover \eq{constr3*} implies that \eq{conf5} is satisfied on $\scri^-$,
i.e.\ it remains to verify that the hypotheses 1.\ and 3.-5.\ in Theorem~\ref{inter-thm} are fulfilled.

Recall that in our gauge the only non-vanishing Christoffel symbols on $\scri^-$ are $\ol \Gamma^k_{ij} = \tilde \Gamma^k_{ij}$,
and that this implies that the action of $\nabla_0$ and $\partial_0$ as well as the action of $\nabla_i$ and $\tilde\nabla_i$  coincides on $\scri^-$.

\subsubsection{Vanishing of $\ol H$, $\ol {\nabla H}$ and $\ol{\nabla\nabla H}$}
\label{sect_vanishing_H}

We have
\begin{eqnarray}
 \ol H^{0}&\equiv& \ol g^{\mu\nu}(\ol\Gamma^0_{\mu\nu} - \hat\Gamma^0_{\mu\nu}) \,=\, 0
\;,
\\
 \ol H^{i}&\equiv& \ol g^{\mu\nu}(\ol\Gamma^i_{\mu\nu} - \hat\Gamma^i_{\mu\nu}) \,=\, 0
\;.
\end{eqnarray}
Equation \eq{cwe5*} can be written as
\begin{equation}
 R_{\mu\nu} - g_{\sigma(\mu}\hat\nabla_{\nu)}H^{\sigma} = 2L_{\mu\nu}
\;.
\label{cwe5_rew}
\end{equation}
Invoking $\ol H^{\sigma}=0$ that gives
\begin{eqnarray*}
 \ol R_{0 0} + \ol{\partial_{0}H^{0}} &=& 2\ol L_{0 0}
\;,
\\
 \ol R_{0i} - \frac{1}{2}\ol g_{ ij}\ol{\partial_{0}H^{j}} &=& 2\ol L_{0i }
\;,
\\
\ol R_{ij} &=& 2\ol L_{ij}
\;.
\end{eqnarray*}
On the other hand, with \eq{constr2*} we find
\begin{eqnarray*}
 \ol R_{00} &=& - \ol{\partial_0 \Gamma^k_{0k}} \,=\, -\frac{1}{2}\ol g^{kl}\ol{\partial_0\partial_0g_{kl}}
\;,
\\
\ol R_{0i} &=&  - \ol{\partial_0 \Gamma^k_{ik}} \,=\, 0
\;,
\\
\ol R_{ij} &=& \ol{\partial_0 \Gamma^0_{ij}} + \tilde R_{ij} \,=\, \frac{1}{2}\ol{\partial_0\partial_0g_{ij}} + \tilde R_{ij}
\;.
\end{eqnarray*}
Taking  \eq{constr5*} into account, we conclude that
\begin{equation}
 \ol{\partial_0\partial_0 g_{ij}} \,=\, 2\tilde R_{ij} - \ol g_{ij}\tilde R
\;,
\label{tt_gij}
\end{equation}
as well as $\ol{\partial_0 H^{\sigma}}=0$, and we end up with
\begin{equation}
 \ol {\nabla_{\mu} H^{\sigma}} \,=\, 0\;.
\end{equation}
Note that this implies
\begin{eqnarray}
 0&=&  \ol{\partial_0 H^0} \,=\, \ol g^{\mu\nu}\ol{\partial_0 \Gamma^0_{\mu\nu}} \,=\, \frac{1}{2}\ol{\partial_0\partial_0g_{00}} - \frac{1}{2}\tilde R
\;,
\\
 0 &=& \ol{\partial_0 H^k} \,=\,   \ol g^{\mu\nu}\ol{\partial_0  \Gamma^k_{\mu\nu}} \,=\, -\ol g^{kl}\ol{\partial_0\partial_0 g_{0l}}
\label{tt_g0i}
\;,
\end{eqnarray}
i.e.\
\begin{eqnarray}
 \ol{\partial_0\partial_0 g_{00}} = \tilde R\;, \quad \ol{\partial_0\partial_0 g_{0i}}=0\;.
\end{eqnarray}
We give a list of the transverse derivatives of the Christoffel symbols on $\scri^-$,
\begin{eqnarray}
& \ol{\partial_0\Gamma^0_{00}} \,=\, -\frac{1}{2}\tilde R\;, \quad \ol{\partial_0\Gamma^0_{ij}} \,=\,
 \ol g_{jk} \ol{\partial_0\Gamma^k_{0i}}
\,=\,\tilde R_{ij} - \frac{1}{2}\ol g_{ij} \tilde R
\;,&
\label{trans_christ_1}
\\
& \ol{\partial_0\Gamma^0_{0i}} \,=\,  \ol{\partial_0\Gamma^k_{00}} \,=\,  \ol{\partial_0\Gamma^k_{ij}} \,=\, 0
\;.&
\label{trans_christ_2}
\end{eqnarray}
Using \eq{cwe5_rew} that yields with $\ol H^{\sigma} =0 = \ol {\nabla_{\mu} H^{\sigma}}$ the relation
\begin{equation}
 \ol{\partial_0 R_{\mu\nu}} - \ol g_{\sigma(\mu}\ol{\partial_{\nu)}\partial_0H^{\sigma}} = 2\ol{\partial_0L_{\mu\nu}}
\;,
\label{cwe5_rew_trans}
\end{equation}
and thus
\begin{eqnarray*}
 \ol{\partial_0 R_{00}} +\ol{\partial_{0}\partial_0H^{0}} &=& 2\ol{\partial_0L_{00}}
 \;,
\\
\ol{\partial_0 R_{0i}} - \frac{1}{2}\ol g_{ ij}\ol{\partial_{0}\partial_0H^{j}} &=& 2\ol{\partial_0L_{0i}}
 \;,
\\
 \ol{\partial_0 R_{ij}} &=& 2\ol{\partial_0L_{ij}}
 \;.
\end{eqnarray*}
We compute
\begin{eqnarray*}
  \ol{\partial_0 R_{00}} &= & - \ol{\partial_{0}\partial_0\Gamma^{k}_{0k}}
\,=\,  -\frac{1}{2}\ol g^{kl} \ol{\partial_0\partial_0\partial_0 g_{kl}}
\;,
\\
 \ol{\partial_0 R_{0i}} &=& \underbrace{ \tilde\nabla_{k}\ol{\partial_0\Gamma^{k}_{0i}}}_{=0} - \ol{\partial_{0}\partial_0\Gamma^{k}_{ik} }
\,=\, -\frac{1}{2}\tilde\nabla_i(\ol g^{kl} \ol{\partial_0\partial_0 g_{kl}}) \,=\, \frac{1}{2}\tilde\nabla_i\tilde R
\;,
\\
 \ol{\partial_0 R_{ij}}  &=&  \ol{\partial_{0}\partial_0\Gamma^{0}_{ij}}
 \,=\, \frac{1}{2}\ol{\partial_0\partial_0\partial_0 g_{ij}}
\;.
\end{eqnarray*}
From \eq{constr6*} we deduce that
\begin{equation}
  \ol{\partial_0\partial_0\partial_0 g_{ij}} \,=\, -4\sqrt{\frac{\lambda}{3}} \, D_{ij}
\;,
\label{ttt_gij}
\end{equation}
from which we obtain $\ol{\partial_0\partial_0 H^{\sigma}}=0$, and thus
\begin{equation}
 \ol {\nabla_{\mu}\nabla_{\nu} H^{\sigma}} \,=\,0
\;.
\end{equation}

\subsubsection{Vanishing of $\ol \zeta$ and $\ol{\nabla\zeta}$}

In our gauge we have
\begin{equation}
 \zeta_{\mu} = -4\nabla_{\alpha}L_{\mu}{}^{\alpha}
\;.
\end{equation}
We invoke
 \eq{constr5*} and \eq{constr6*} to obtain
\begin{eqnarray}
 \ol \zeta_0 &=& 4\ol{\partial_{0}L_{00}}  -4\ol g^{kl} \nabla_{k}\ol L_{0l} \,=\, 0
\;,
\\
\ol\zeta_i &=& 4\ol{\partial_{0}L_{0i}} -4\ol g^{kl}\tilde\nabla_{k}\ol L_{il} \,=\, 0
\;.
\end{eqnarray}
The computation of $\ol{\nabla_{0}\zeta_{\mu}}$ requires the knowledge of certain second-order transverse derivatives of $\ol L_{\mu\nu}$ which we compute
from the CWE \eq{cwe1*}. Since $\ol H^{\sigma} =0=\ol  {\nabla_{\mu} H^{\sigma}}$ we have
\begin{eqnarray*}
& \ol{\Box_g L_{\mu\nu}} =  \ol{\Box^{(H)}_g L_{\mu\nu}} = 4\ol  L_{\mu\kappa}\ol  L_{\nu}{}^{\kappa} - \ol  g_{\mu\nu}| \ol L|^2 \quad \Longleftrightarrow&
\\
&\ol{\nabla_0\nabla_0  L_{\mu\nu}}  = \Delta_{\tilde g}\ol L_{\mu\nu}-  4\ol  L_{\mu\kappa}\ol  L_{\nu}{}^{\kappa} +\ol  g_{\mu\nu}[\ol L_{k}{}^l\ol L_l{}^k + (\ol L_{00})^2]
\;,&
\end{eqnarray*}
whence
\begin{eqnarray}
 \ol{\nabla_0\nabla_0  L_{00}} & =& \frac{1}{4}\Delta_{\tilde g}\tilde R - |\tilde R|^2  + \frac{1}{2}\tilde R^2
\;,
\\
\ol{\nabla_0\nabla_0  L_{0i}}  &=&0
\;.
\end{eqnarray}
From
\begin{eqnarray*}
 \ol{ \nabla_0\zeta_{\mu}} &=&  4  \ol{\nabla_{0}\nabla_0L_{0\mu}}  -4 \ol g^{kl} \tilde\nabla_{k}\ol{\nabla_0L_{l\mu}} -4  \ol R_{0k\mu}{}^{l}\ol L_{l}{}^{k}
 - 4 \ol R_{00}\ol L_{0\mu}
\end{eqnarray*}
and
\begin{eqnarray}
 \ol R_{0k0}{}^l &=& -\tilde R_{k}{}^l + \frac{1}{2}\delta_k{}^l\tilde R
\;,
\label{riemann1}
\\
 \ol R_{0ki}{}^l  &=& 0
\;,
\label{riemann2}
\end{eqnarray}
we conclude that
\begin{eqnarray}
 \ol{ \nabla_0\zeta_{0}} &=&  \frac{3}{2}\tilde R^2 -4 |\tilde R|^2   -4  \ol R_{0k0}{}^{l}\ol L_{l}{}^{k}
 \,=\,   0
\;,
\\
 \ol{ \nabla_0\zeta_{i}} &=&    4 \sqrt{\frac{\lambda}{3}}  \tilde\nabla^{j}D_{ij}-4  \ol R_{0ki}{}^{l}\ol L_{l}{}^{k} \,=\, 0
\;.
\end{eqnarray}

\subsubsection{Validity of the MCFE \eq{conf1}-\eq{conf4} and their transverse derivatives on $\scri^-$}

The independent components of $ \ol{\nabla_{\rho} d_{\mu\nu\sigma}{}^{\rho}}$, which is antisymmetric in its first two indices, trace-free and satisfies the
first Bianchi identity, are
\begin{equation*}
   \ol{\nabla_{\rho} d_{ijk}{}^{\rho}} \quad \text{and} \quad    \ol{\nabla_{\rho} d_{0ij}{}^{\rho}}
\end{equation*}
(similarly for its transverse derivatives).

It follows from \eq{spatial_tr}, \eq{constr7*}, \eq{constr8*} and \eq{constr1*} that
\begin{eqnarray}
  \ol{\nabla_{\rho} d_{ijk}{}^{\rho}} &=&   \tilde \nabla^{l}\ol  d_{ijkl} -   \ol{\nabla_{0} d_{0kji}}
\,=\,  2\ol g_{k[i}\tilde \nabla^{l}D_{j]l}
\,=\, 0
\;,
\\
 \ol{\nabla_{\rho} d_{0ij}{}^{\rho}}  &=&   \nabla^{k} \ol d_{0ijk} +  \ol{\nabla_{0} d_{0i0j}}
\,=\, 0
\;.
\end{eqnarray}

We consider the corresponding transverse derivatives. With \eq{riemann1} and \eq{riemann2} we find
\begin{eqnarray*}
  \ol{\nabla_0\nabla_{\rho} d_{ijk}{}^{\rho}} &=& \ol{\nabla_{0}\nabla_0 d_{0kij}} + \tilde\nabla^{l}\ol{\nabla_0 d_{ijkl}}
- 2\ol R_{0[ j|0}{}^l\ol d_{0|i] kl}  + \ol R_{0k0}{}^{l } \ol d_{0lij}  + \ol R_{00}\ol d_{0kji}
\;,
\\
 \ol{\nabla_0\nabla_{\rho} d_{0ij}{}^{\rho}}  &=& \ol{\nabla_{0}\nabla_0 d_{0i0j}} +  \tilde \nabla^{k}\ol{\nabla_0 d_{0ijk}}
- \ol R_{0}{}^{k}{}_{ 0}{}^{l} \ol d_{ik jl}
+\ol  R_{0j0}{}^{k} D_{ik} -\ol  R_{0 0}  D_{ij}
\;.
\end{eqnarray*}
The second-order transverse derivatives of the rescaled Weyl tensor follow from the CWE \eq{cwe4*},
\begin{eqnarray*}
 \ol{\Box_g d_{\mu\nu\sigma\rho} }\,=\,   \ol{ \Box^{(H)}_g d_{\mu\nu\sigma\rho}}
  &=& 0 \quad \Longleftrightarrow \quad
\ol{\nabla_0\nabla_0 d_{\mu\nu\sigma\rho}} \,=\, \Delta_{\tilde g}  \ol d_{\mu\nu\sigma\rho}
\;,
\end{eqnarray*}
hence
\begin{eqnarray}
 \ol{\nabla_0\nabla_0 d_{0ijk}} &=& \Delta_{\tilde g}  \ol d_{0ijk} \,=\, \sqrt{\frac{3}{\lambda}}  \Delta_{\tilde g} \tilde C_{ijk}
\;,
\\
  \ol{\nabla_0\nabla_0 d_{0i0j}} &=& \Delta_{\tilde g}  \ol d_{0i0j} \,=\, \Delta_{\tilde g}  D_{ij}
\;.
\end{eqnarray}
The Bianchi identities together with the identity
\begin{eqnarray}
 \tilde R_{ijkl} &\equiv &  2 \ol g_{i[k}\tilde R_{l]j} - 2\ol g_{j[k}\tilde R_{l]i}   - \tilde R \ol g_{i[k}\ol g_{l]j}
\;,
\end{eqnarray}
which holds in 3 dimensions,
imply the following relations for Cotton and Bach tensor,
\begin{eqnarray*}
\tilde C_{[ijk]} &=&  \tilde C^j{}_{ij} \,=\, \tilde\nabla^k\tilde C_{kij}\,=\, 0
 \;,
\\
  \tilde\nabla_{[i}\tilde C_{j]kl}
  &=& \tilde\nabla_{[l}\tilde C_{k]ji} +\tilde R_{ij[l}{}^m \tilde L_{k]m} + \tilde R_{kl[i}{}^m \tilde L_{j]m}
 \;,
\\
 \tilde\nabla^j \tilde B_{ij} &=& \tilde R^{kl}\tilde C_{kli}
 \;,
 \\
 \tilde\nabla_{[i}\tilde B_{j]k}
  &=&  - \frac{1}{2}\Delta_{\tilde g}\tilde C_{kji}  +  \tilde R_{[j}{}^l \tilde C_{i]kl}
   -\frac{1}{2} \tilde R_{k}{}^{l} \tilde C_{lij}
-g_{k[i}\tilde C^l{}_{j]}{}^{m}  \tilde R_{lm}
 + \frac{1}{4}\tilde R  \tilde C_{kji}
 \;.
\end{eqnarray*}
With \eq{spatial_tr}, \eq{constr1*}, \eq{constr7*} and \eq{constr8*} we then obtain
\begin{eqnarray*}
  \ol{\nabla_0\nabla_{\rho} d_{ijk}{}^{\rho}} &=&
    \sqrt{\frac{3}{\lambda}}\Big( 2\ol g_{k[i}\tilde\nabla^{l} \tilde B_{j]l}
    -  2\tilde\nabla_{[i}\tilde B_{j]k}
     -  (\Delta_{\tilde g}   - \frac{\tilde R  }{2}  ) \tilde C_{kji}
+ 2   \tilde R_{[ j}{}^l \tilde C_{i]kl}- \tilde R_{k}{}^{l } \tilde C_{lij}
  \Big)
\\
 &=& 0
\;,
\\
 \ol{\nabla_0\nabla_{\rho} d_{0ij}{}^{\rho}}  &=& 0
\;.
\end{eqnarray*}


Set
\begin{equation}
 \Xi_{\mu\nu} :=  \nabla_{\mu}\nabla_{\nu}\Theta + \Theta L_{\mu\nu} - s g_{\mu\nu}
 \;.
\end{equation}
To compute $\ol \Xi_{00}$ we need to know the value of $\ol{\nabla_0\nabla_0 \Theta}$ which can be determined from the CWE \eq{cwe3*},
\begin{equation}
 \ol{\Box_g \Theta} = 4\ol s \quad \Longleftrightarrow \quad \ol{\nabla_0\nabla_0\Theta}=0
\;.
\label{tt-Theta}
\end{equation}
Invoking  \eq{constr3*}-\eq{constr4*} we then find
\begin{eqnarray}
  \ol\Xi_{ij} &=&0
 \;,
\\
  \ol\Xi_{0i} &=& \ol{ \nabla_{i}\nabla_{0}\Theta} \,=\, 0
 \;,
\\
  \ol\Xi_{00} &=&  \ol{\nabla_{0}\nabla_{0}\Theta} \,=\, 0
 \;.
\end{eqnarray}
To calculate the transverse derivative of $\Xi_{\mu\nu}$ on $\scri^-$ we need to determine the third-order transverse derivative of $\Theta$
\begin{equation}
 \ol{\nabla_0\Box_g \Theta} = 4\ol{\nabla_0 s} \quad \Longleftrightarrow \quad  \ol{\nabla_{0}\nabla_0 \nabla_0\Theta}
 =- \sqrt{\frac{\lambda}{12}}\tilde R
\;.
\end{equation}
One then straightforwardly verifies with \eq{riemann1} and the constraint equations
\begin{eqnarray*}
  \ol{\nabla_0\Xi_{ij}} &=& \ol{\nabla_{i}\nabla_{j} \nabla_0\Theta} + \ol R_{0i 0j}\ol{\nabla_0\Theta}+ \ol L_{ij}\ol{\nabla_0 \Theta}  -\ol{ \nabla_0s}\ol  g_{ij}
 \,=\,0
\;,
\\
  \ol{\nabla_0\Xi_{0i}} &=& \ol{ \nabla_{i}\nabla_{0}\nabla_0\Theta} + \ol L_{0i}\ol{\nabla_0 \Theta} \,=\, 0
\;,
\\
  \ol{\nabla_0\Xi_{00}} &=& \ol{\nabla_{0}\nabla_{0} \nabla_0\Theta}+ \ol L_{00}\ol{\nabla_0 \Theta}  +\ol{ \nabla_0s} \,=\, 0
\;.
\end{eqnarray*}


Set
\begin{equation}
 \Upsilon_{\mu} :=  \nabla_{\mu} s + L_{\mu\nu}\nabla^{\nu}\Theta
 \;.
\end{equation}
We observe that by \eq{constr3*}-\eq{constr5*}
\begin{eqnarray}
\ol \Upsilon_{0} &=&  \ol{\nabla_{0} s} - \ol L_{00}\ol{\nabla_0 \Theta} \,=\, 0
 \;,
\\
 \ol \Upsilon_{i} &=&  0
 \;.
\end{eqnarray}
To compute the corresponding transverse derivatives on $\scri^-$ we first of all need to calculate  $\ol{\nabla_0\nabla_0 s}$, which follows from
\eq{cwe2*},
\begin{equation}
 \ol{\overset{}{\Box_g} s} =0 \quad \Longleftrightarrow \quad \ol{\nabla_0\nabla_0s}=0
;.
\end{equation}
Employing further the constraint equations and \eq{tt-Theta} we then deduce
\begin{eqnarray*}
 \ol{\nabla_0 \Upsilon_{0}} &=&  \ol{\nabla_0\nabla_{0} s} - \ol{\nabla_0L_{00}\nabla_{0}\Theta}
  -\ol{ L_{00}\nabla_{0}\nabla_0\Theta} \,=\,0
\;,
\\
 \ol{\nabla_0 \Upsilon_{i}} &=&  \nabla_{i}\ol{\nabla_0 s} - \ol{\nabla_0L_{0i}\nabla_{0}\Theta}
  +\ol L_{ij}\nabla^{j}\ol{\nabla_0\Theta} \,=\,0
\;.
\end{eqnarray*}


Set
\begin{equation}
 \varkappa_{\mu\nu\sigma} :=  2\nabla_{[\sigma} L_{\nu]\mu} -\nabla_{\rho}\Theta \, d_{\nu\sigma\mu}{}^{\rho}
 \;.
\end{equation}
Due to the symmetries $\varkappa_{\mu(\nu\sigma)}=0$, $\varkappa_{[\mu\nu\sigma]}=0$ and $\varkappa_{\nu\mu}{}^{\nu}=0$ (since $\overline\zeta_{\mu}=0$ and $L=0$)
its independent components on the initial surface are
\begin{equation*}
 \ol\varkappa_{ijk} \quad \text{and} \quad   \ol\varkappa_{ij0}
\;.
\end{equation*}
Since also $\overline{\nabla_0 \zeta_{\mu}}=0$ an analogous statement holds true for $\ol{\nabla_0\varkappa_{\mu\nu\sigma}}$.
We find with \eq{constr3*} and \eq{constr5*}-\eq{constr7*}
\begin{eqnarray}
 \ol \varkappa_{ijk} &=&  2\nabla_{[k} \ol L_{j]i} -\ol{\nabla_{0}\Theta} \, \ol d_{0ijk}
\,=\,  0
 \;,
\\
 \ol\varkappa_{ij0} &=&  2\ol{\nabla_{[0} L_{j]i} }+\ol{\nabla_{0}\Theta} \, \ol d_{0i0j}
\,=\,  0
 \;.
\end{eqnarray}
Before we proceed let us first determine the second-order transverse derivative of $L_{ij}$ on $\scri^-$. From the CWE
\eq{cwe1*} we obtain
\begin{eqnarray}
  \ol{\Box_gL_{ij}} =   \ol{\Box^{(H)}_gL_{ij}} 4\tilde L_{ik}\tilde L_{j}{}^k - \ol g_{ij}|\tilde L|^2 - \ol g_{ij} (\ol L_{00})^2 \quad
\Longleftrightarrow
\nonumber
\\
 \ol{\nabla_0\nabla_0L_{ij}} = \Delta_{\tilde g} \tilde L_{ij}  -4\tilde L_{ik}\tilde L_{j}{}^k + \ol g_{ij}( |\tilde L|^2 + \frac{1}{16} \tilde R^2)
\;.
\end{eqnarray}
For the transverse derivatives we then find with \eq{riemann1}, \eq{riemann2} and \eq{tt-Theta} and the constraint equations
\begin{eqnarray*}
  \ol{\nabla_0\varkappa_{ijk} }&=&
 2\tilde \nabla_{[k}\ol{\nabla_{|0|} L_{j]i} }  -\ol{\nabla_{0}\Theta} \, \ol{\nabla_0d_{0ijk}}
\,=\,0
\;,
\\
  \ol{\nabla_0\varkappa_{ij0} }&=&
\ol{ \nabla_{0}\nabla_0 L_{ij}}  - \tilde \nabla_{j}\ol{\nabla_0 L_{0i}}  -\ol R_{0j 0}{}^{k}\ol  L_{ i k} - \ol R_{0i0j} \ol L_{00}
 +\ol{\nabla_{0}\Theta} \, \ol{\nabla_0d_{0i0j}}
\\
&=&
\ol{ \nabla_{0}\nabla_0 L_{ij}}  -\frac{1}{4}\tilde\nabla_i \tilde \nabla_{j}\tilde R
+\tilde L_{j }{}^{k}\tilde  L_{ i k}     - \frac{1}{16}\tilde R^2 \ol g_{ij}
 +\tilde B_{ij}
\,=\,0
\;,
\end{eqnarray*}
where we have used that
\begin{equation}
\tilde B_{ij} = -\Delta_{\tilde g}\tilde L_{ij} + \frac{1}{4}\tilde\nabla_i \tilde \nabla_{j}\tilde R
 -   \ol g_{ij}|\tilde L|^2
+ 3\tilde L_{ik}\tilde L_{j}{}^k
\;.
\end{equation}

\subsubsection{Vanishing of $\ol W_{\mu\nu\sigma}{}^{\rho}-\ol\Theta\, \ol d_{\mu\nu\sigma}{}^{\rho}$
 and $\ol{\nabla_0 (W_{\mu\nu\sigma}{}^{\rho} - \Theta \, d_{\mu\nu\sigma}{}^{\rho}) }$}

The independent components of the conformal Weyl tensor in adapted coordinates are
\begin{equation*}
 \ol W_{0ij}{}^{k}\quad \text{and} \quad \ol W_{0i0}{}^{j}
\;.
\end{equation*}
Using the definition of the Weyl tensor
\begin{equation*}
 W_{\mu\nu\sigma}{}^{\rho} \equiv  R_{\mu\nu\sigma}{}^{\rho} -  2\left(g_{\sigma[\mu} L_{\nu]}{}^{\rho}  - \delta_{[\mu}{}^{\rho}L_{\nu]\sigma} \right)
\end{equation*}
we observe that by \eq{riemann1}, \eq{riemann2} and \eq{constr5*} we have
\begin{eqnarray*}
  \ol W_{0ij}{}^{k} &=&  \ol g_{ij} \ol L_{0}{}^{k}  - \delta_{i}{}^{k}\ol L_{0j} \,=\,0
\;,
\\
\ol W_{0i0}{}^{j} &=&   \ol R_{0i0}{}^{j}+ \ol  L_{i}{}^{j}  - \delta_{i}{}^{j}\ol L_{00} \,=\, 0
 \;.
\end{eqnarray*}

To derive  expressions for the transverse derivatives recall the formulae \eq{christoffel}, \eq{trans_christ_1}-\eq{trans_christ_2} for the Christoffel symbols and their transverse derivatives on~$\scri^-$.
Since, by \eq{tt_gij}, \eq{tt_g0i} and \eq{ttt_gij}, we further have
\begin{eqnarray*}
 \ol{\partial_0\partial_0\Gamma^k_{ij}} &=& \frac{1}{2}\ol g^{kl}(\tilde\nabla_i\ol{\partial_0\partial_0g_{jl}} + \tilde\nabla_j\ol{\partial_0\partial_0g_{il}}
 - \tilde\nabla_l\ol{\partial_0\partial_0g_{ij}})
\\
 &=& 2\tilde\nabla_{(i}\tilde R_{j)}{}^k  - \delta_{(i}{}^k \tilde\nabla_{j)}\tilde R
 - \tilde\nabla^k\tilde R_{ij} + \frac{1}{2} \ol g_{ij} \tilde\nabla^k\tilde R
 \;,
\\
\ol{\partial_0\partial_0\Gamma^j_{i0} } &=& \frac{1}{2}\ol g^{jk}(\ol{\partial_0\partial_0\partial_0g_{ik}} + \tilde\nabla_i\ol{\partial_0\partial_0 g_{0k}}
 -  \tilde\nabla_k\ol{\partial_0\partial_0 g_{0i}})
\,=\, -2\sqrt{\frac{\lambda}{3}} D_i{}^j
\;,
\end{eqnarray*}
we find that
\begin{eqnarray*}
 \ol{\nabla_0{R_{0ij}{}^{k}}} \,=\,   \ol{\partial_0{R_{0ij}{}^{k}}}
&=& \tilde\nabla_i\ol{\partial_0\Gamma^k_{0j}} - \ol{\partial_0\partial_0\Gamma^k_{ij} }
\\
&=& -\tilde\nabla_j\tilde R_i{}^k + \frac{1}{2}\delta_i{}^k\tilde\nabla_j\tilde R
 + \tilde\nabla^k\tilde R_{ij} - \frac{1}{2} \ol g_{ij} \tilde\nabla^k\tilde R
\;,
\\
\ol{\nabla_0{R_{0i 0}{}^{j}}}  \,=\,    \ol{\partial_0{R_{0i 0}{}^{j}}}
   &=&  - \ol{\partial_0\partial_0\Gamma^j_{i0} }
\,=\,2\sqrt{\frac{\lambda}{3}} D_i{}^j
\;.
\end{eqnarray*}
Hence
\begin{eqnarray*}
 \ol{\nabla_0 W_{0ij}{}^{k}} &=&  \ol{\nabla_0{R_{0ij}{}^{k}}} + \ol g_{ij} \ol {\nabla_0 L_{0}{}^{k}}
 - \delta_{i}{}^{k}\ol {\nabla_0 L_{0j}}
\,=\,   \tilde\nabla^k\tilde L_{ij}-\tilde\nabla_j\tilde L_i{}^k  \,=\, \tilde C_{ij}{}^k
\;,
\\
 \ol{\nabla_0 W_{0i 0}{}^{j}} &=&  \ol{\nabla_0{R_{0i 0}{}^{j}}} +  \ol {\nabla_0 L_{i}{}^{j}}
- \delta_{i}{}^{j}\ol {\nabla_0 L_{00}}
\,=\,  \sqrt{\frac{\lambda}{3}} \, D_{i}{}^j
\;,
\end{eqnarray*}
and we end up with
\begin{eqnarray*}
  \ol {\nabla_0(W_{0ij}{}^{k}- \Theta  d_{0ij}{}^{k})} &=&
   \ol {\nabla_0W_{0ij}{}^{k}} -\ol{\nabla_0\Theta}\, \ol d_{0ij}{}^{k}\,=\, 0
\;,
\\
  \ol {\nabla_0(W_{0i0}{}^{j}- \Theta d_{0i0}{}^{j})} &=&
   \ol {\nabla_0W_{0i0}{}^{j}} - \ol{\nabla_0\Theta}\, \ol d_{0i0}{}^{j}\,=\, 0
\;,
\end{eqnarray*}
which completes the proof that Theorem~\ref{inter-thm} is applicable supposing that the initial data for the CWE satisfy the constraint equations
\eq{constr1*}-\eq{constr8*} on $\scri^-$.

\begin{theorem}
\label{thm_equivalence}
 Let us suppose we have been given a Riemannian metric $h_{ij}$ and a smooth tensor field $D_{ij}$ on $\scri^-$.
A smooth solution $(g_{\mu\nu},L_{\mu\nu},d_{\mu\nu\sigma}{}^{\rho},\Theta,s)$ of the CWE \eq{cwe1*}-\eq{cwe5*} to the future of $\scri^-$ with initial data
\begin{eqnarray*}
(\ol g_{\mu\nu} = \mathring g_{\mu\nu}, \enspace\ol{\partial_0 g_{\mu\nu}}=\mathring K_{\mu\nu},\enspace\ol L_{\mu\nu}=\mathring L_{\mu\nu},\enspace\ol{\partial_0 L_{\mu\nu}}=\mathring M_{\mu\nu}, \enspace
\ol d_{\mu\nu\sigma}{}^{\rho}= \mathring d_{\mu\nu\sigma}{}^{\rho},
\\
\ol{\partial_0 d_{\mu\nu\sigma}{}^{\rho}}=\mathring D_{\mu\nu\sigma}{}^{\rho},\enspace\ol\Theta = \mathring \Theta=0,\enspace\ol{\partial_0\Theta}=\mathring\Omega,\enspace\ol{s}=\mathring s=0,\ol{\partial_0 s}=\mathring S)
\end{eqnarray*}
where
$\mathring g_{ij}=h_{ij}$ and the trace- and divergence-free part of $\mathring d_{0i0j}=D_{ij}$  are the free data, is a solution of the MCFE \eq{conf1}-\eq{conf6} in the
$$(R=0, \ol s=0,  \ol g_{00}=-1, \ol g_{0i}=0, \hat g_{\mu\nu} = \mathring g_{\mu\nu})\text{-wave-map gauge}$$
 if and only if the initial data have their usual algebraic properties and solve the constraint equations \eq{constr1*}-\eq{constr8*}.
The function $\Theta$ is positive in some neighborhood to the future of $\scri^-$, and $\mathrm{d}\Theta \ne 0$ on $\scri^-$.
\end{theorem}

\end{document}